\documentclass[twocolumn,10pt]{IEEEtran}

\usepackage{calc}

\usepackage{xcolor,multirow,gensymb}
  \usepackage{cite}
\usepackage{graphicx,subfigure,epsfig}
\usepackage{psfrag}
\usepackage{amsmath,amssymb}
\usepackage{color}
\usepackage{array}
\usepackage{nicefrac}
\usepackage{amsmath,amssymb}
\usepackage{xifthen}
\usepackage{mathtools}
\usepackage{enumerate}
\usepackage{microtype}
\usepackage{xspace}
\usepackage{bm}
\usepackage[T1]{fontenc}
\usepackage{fancyhdr}
\usepackage{lastpage}
\usepackage{rotating}
\usepackage{tabulary}

\newcommand{\vect}[1]{{\lowercase{\mbs{#1}}}}

\newcommand{\mbs}[1]{\bm{#1}}
\newcommand{\mat}[1]{{\uppercase{\mbs{#1}}}}
\newcommand{\Id}{\mat{\mathrm{I}}}

\newcommand{\T}{{\scriptscriptstyle\mathsf{T}}}
\renewcommand{\H}{{\scriptscriptstyle\mathsf{H}}}

\renewcommand{\Re}[1][]{\ifthenelse{\isempty{#1}}{\operatorname{Re}}{\operatorname{Re}\left(#1\right)}}
\renewcommand{\Im}[1][]{\ifthenelse{\isempty{#1}}{\operatorname{Im}}{\operatorname{Im}\left(#1\right)}}
\usepackage{float}

\newfloat{longequation}{b}{ext}

\newcommand{\etv}{\vect{\eta}}

\def\bH{{\mathbf{H}}}
\def\bI{{\mathbf{I}}}

\def\bN{{\mathbf{N}}}

\def\bQ{{\mathbf{Q}}}

\def\bW{{\mathbf{W}}}

\newcommand{\cC}{{\cal C}}

\newcommand{\cN}{{\cal N}}

\def\bb{{\mathbf{b}}}

\def\bd{{\mathbf{d}}}
\def\bee{{\mathbf{e}}}

\def\bg{{\mathbf{g}}}
\def\bh{{\mathbf{h}}}

\def\bs{{\mathbf{s}}}

\def\bw{{\mathbf{w}}}

\def\by{{\mathbf{y}}}

\def\b0{{\mathbf{0}}}

\def\bbC{{\mathbb{C}}}

\def\bpsi{{\boldsymbol{\psi}}}

\newcommand{\EE}{\mathbb{E}}

\newcommand{\CN}[1][]{\ifthenelse{\isempty{#1}}{\mathcal{N}_{\mathbb{C}}}{\mathcal{N}_{\mathbb{C}}\left(#1\right)}}

\renewcommand{\P}[1][]{\ifthenelse{\isempty{#1}}{\mathbb{P}}{\mathbb{P}\left(#1\right)}}
\renewcommand{\det}[1][]{\ifthenelse{\isempty{#1}}{\text{det}}{\text{det}\left(#1\right)}}
\newcommand{\trace}[1][]{\ifthenelse{\isempty{#1}}{\text{tr}}{\text{tr}\left(#1\right)}}
\newcommand{\rank}[1][]{\ifthenelse{\isempty{#1}}{\text{rank}}{\text{rank}\left(#1\right)}}
\newcommand{\diag}[1][]{\ifthenelse{\isempty{#1}}{\text{diag}}{\text{diag}\left(#1\right)}}
\def\nn{\nonumber}

\IEEEoverridecommandlockouts

\newtheorem{proposition}{Proposition}
\newtheorem{remark}{Remark}
\newtheorem{definition}{Definition}
\newtheorem{theorem}{Theorem}
\newtheorem{corollary}{Corollary}
\newtheorem{lemma}{Lemma}

\newcommand{\tr}{\mathop{\mathrm{tr}}\nolimits}

\newcommand{\al}{\alpha}

\newcounter{enumi_saved}
\usepackage{answers}
\Newassociation{solution}{Solution}{solutionfile}

\AtBeginDocument{\Opensolutionfile{solutionfile}[\jobname]}
\AtEndDocument{\Closesolutionfile{solutionfile}\clearpage
}

\usepackage{dsfont}
\usepackage{cool}
\usepackage{minibox}
\DeclareSymbolFont{matha}{OML}{txmi}{m}{it}
\DeclareMathSymbol{\varv}{\mathord}{matha}{118}
\usepackage{eurosym}

\makeatletter

\newcounter{eqnback}

\IEEEoverridecommandlockouts
\begin{document}

\title{Towards a Realistic Assessment of Multiple Antenna HCNs: Residual Additive Transceiver Hardware Impairments and Channel Aging}
\author{Anastasios Papazafeiropoulos  and Tharm Ratnarajah   \vspace{2mm} \\
\thanks{Copyright (c) 2015 IEEE. Personal use of this material is permitted. However, permission to use this material for any other purposes must be obtained from the IEEE by sending a request to pubs-permissions@ieee.org.}
\thanks{A. Papazafeiropoulos and T. Ratnarajah are  with the  Institute for Digital Communications (IDCOM), University of Edinburgh, Edinburgh, EH9 3JL, U.K., (email: {a.papazafeiropoulos, t.ratnarajah}@ed.ac.uk). }
\thanks{This work was supported by the U.K. Engineering and Physical Sciences Research Council (EPSRC) under grants EP/N014073/1 and EP/N015312/1.}}
\maketitle

\vspace{-1cm}

\begin{abstract}
Given the critical dependence of broadcast channels by the accuracy of channel state information at the transmitter (CSIT), we develop a general downlink model with zero-forcing (ZF) precoding, applied in realistic  heterogeneous cellular systems with multiple antenna base stations (BSs). Specifically, we take into consideration imperfect CSIT due to pilot contamination, channel aging due to users relative movement, and unavoidable residual additive  transceiver hardware impairments (RATHIs). Assuming that the BSs are Poisson distributed, the main contributions focus on the derivations of the upper bound of the coverage probability and the achievable user rate for this general model. We show that both the coverage probability and the user rate are dependent on the imperfect CSIT and RATHIs. More concretely, we  quantify the resultant performance loss of the network  due to these effects. We depict that the uplink RATHIs have equal impact, but the downlink transmit BS distortion has a greater impact than the receive  hardware impairment of the user. Thus, the transmit BS hardware should be of better quality than user's receive hardware.  Furthermore, we characterise  both the  coverage probability and user rate in terms of the time variation of the channel.  It is shown that both of them decrease with increasing user mobility, but after a specific value of the normalised Doppler shift, they increase again. Actually, the time variation,  following the Jakes autocorrelation function, mirrors this effect on coverage probability and user rate. Finally, we consider space division multiple access (SDMA), single user beamforming (SU-BF), and baseline single-input single-output (SISO) transmission. A comparison among these schemes reveals that the coverage by means of SU-BF outperforms SDMA in terms of coverage.
\end{abstract}
\begin{keywords}
Channel estimation, channel aging, additive hardware impairments,  coverage probability, multiple antenna  heterogeneous cellular networks.
\end{keywords}

 \section{Introduction}
The design of the future Fifth Generation (5G) networks, demanding to cover the upcoming avalanche of wireless traffic volume due to the accompanied societal development, is quite challenging.  Intercell interference is considered as a key limiting factor among the next generation of wireless systems~\cite{Thompson2014,Wong2017}, which include both a vector Gaussian broadcast and interference channels by means of multi-user multiple-input multiple-output (MU-MIMO) and multi-cell scenarios, respectively. Simplistic models such as the Wyner model~\cite{Shamai1997}, where the intercell interference is assumed to be a constant factor of the total interference, are highly inaccurate given the dramatic variation of the signal-to-interference-plus-noise ratio (SINR) value across a cell. Nevertheless, other works consider a fixed user or a small number of interfering BSs which might provide tractable results but very pessimistic with not much insight on users' performance~\cite{Goldsmith2005}. 
 
Fortunately, tractable and accurate models have been developed for studying the downlink coverage, taking into account the full network interference~\cite{Andrews2011,Dhillon2012,Jo2012}. Specifically, the introduction of the randomness regarding the locations of the BSs by means of a heterogeneous Poisson point process (PPP), which allows the usage of tools from stochastic geometry, has taken place. In this technique, known as heterogeneous networks (HetNets) design, small cells are embodied into a macrocell network with main benefits being a dense coverage and ubiquitous high throughput~\cite{Andrews2012}. Hence, novel results regarding the coverage probability have been derived to quantify the quality of service in next generation networks without the need for Monte-Carlo simulations. HetNets design belongs to the major technologies currently on the table for 5G. The heterogeneous cellular networks (HCNs), which are the focal point of this study, can be considered as HetNets of a single tier.

Although appealing in their concept, HCNs as any other network, are  hampered by the inevitable transceiver hardware impairments~\cite{Schenk2008,Studer2010,Goransson2008,Bjornson2012Optimal,Bjoernson2013,Qi2012,Mehrpouyan2012,Papazafeiropoulos2015b,PapazafeiropoulosMay2016,Papazafeiropoulos2017b,PapazafeiropoulosJuly2016,PapazafeiropoulosJuly2017,Papazafeiropoulos2016b,Papazafeiropoulos2016,Papazafeiropoulos2015c,Papazafeiropoulos2017a,Papazafeiropoulos2017}. The impact of hardware impairments is a major challenge because the applied compensation algorithms, which include analog and digital signal processing, cannot remove the impairments completely, since the time-varying hardware characteristics cannot be fully parameterized and estimated, and because there is a randomness induced by different sources of noise~\cite{Schenk2008,Bjornson2015,Zhu2017,Papazafeiropoulos2016}. Especially, cheap hardware components, being attractive for industrial implementation, are particularly prone to the transceiver impairments. Such impairments are originated by amplifier non-linearities, I/Q-imbalance, and quantization errors~\cite{Schenk2008,Studer2010,Qi2012}, and can mainly be modelled as residual additive impairments~\cite{Schenk2008,Studer2010,Goransson2008,Bjornson2012Optimal,Bjoernson2013,Papazafeiropoulos2015b,Papazafeiropoulos2017b,PapazafeiropoulosMay2016,PapazafeiropoulosJuly2016,Papazafeiropoulos2016b,PapazafeiropoulosJuly2017,Papazafeiropoulos2016,Papazafeiropoulos2015c,Papazafeiropoulos2017a,Papazafeiropoulos2017}. In the literature, three basic categories are met, namely, the residual additive impairments, the multiplicative impairments, and the amplified thermal noise \cite{Bjornson2015,Zhu2017}. The residual additive impairments, modelled as independent additive distortion noises at the BS as well as at the user, describe the aggregate effect from many impairments. On the other hand, there are hardware impairments multiplied with the channel vector, which might cause channel attenuation and phase shifts.  They appear as a multiplicative distortion that cannot be incorporated by the channel vector. Herein, our analysis  focuses on the impact of the residual additive hardware impairment, while the study of multiplicative impairments is left for future work. Regarding the introduction of amplified thermal noise, this is straightforward. Moreover, it is worthwhile to mention that the adoption of the current model for the additive impairments is based on its analytical tractability and the experimental verifications. Remarkably, some of the authors have achieved to introduce the rate-splitting approach as a robust method against the RATHIs, although these impairments are residual \cite{Papazafeiropoulos2017,Papazafeiropoulos2017a}. The topic of dealing with other methods and strategies to mitigate the RATHIs is left for future work. Notably, since HetsNets are a candidate solution for 5G systems~\cite{Andrews2014,Zheng2015}, they need to be evaluated in the presence of the detrimental transceiver hardware imperfections, in order to make realistic conclusions regarding their final implementation. It is worthwhile to mention that the research area of HCNs has not considered this kind of inevitable degradations except~\cite{Papazafeiropoulos2016b}, which in turn, brings out a gold rush for new research.  Note that~\cite{Papazafeiropoulos2016b} assumes perfect channel state information at the transmitter (CSIT), neglecting that, in practice, CSIT is imperfect due to several reasons.

Several works have investigated the impact of generally imperfect CSIT in different scenarios of HetNets \cite{Li2015a,Geraci2014,Li2015,Yang2016,Dhungana2017}. However, except \cite{Bai2016}, no other previous work has taken into account for pilot contamination due to the re-use of the pilot sequences or users' mobility, which both are inevitable degrading causes on system's performance and result in imperfect CSIT~\cite{Marzetta2010,Truong2013,Papazafeiropoulos2015a}. In particular, pilot contamination is an inherent weakness in next generation systems aiming to employ the concept of massive MIMO which include a  large number of antennas (tens or hundreds) at each BS~\cite{Marzetta2010}. Notably, massive MIMO are based on time-division duplexing (TDD) mode for channel estimation. With respect to users' mobility and its induced time variation of the channel,  a lack of investigation appears in the literature~\cite{Truong2013,Papazafeiropoulos2015a}. Its practical importance becomes higher, especially, in outdoor urban environments, where the mobility of terminals is increased. Despite that numerous studies, concerning 5G systems, have been presented (see \cite{5g} and references therein), the lack of channel aging works continues. The impact of terminal mobility is often modelled by means of a stationary ergodic Gauss-Markov block fading \cite{Baddour2005,Truong2013,Papazafeiropoulos2015a,Papazafeiropoulos2014WCNC,Papazafeiropoulos2014a,Kong2015,Kong2015a,Papazafeiropoulos2014a,Papazafeiropoulos2016a}. For example, in~\cite{Truong2013 }, the authors provided deterministic equivalents (DEs) for the maximal-ratio-combining (MRC) receivers in the uplink and the maximal-ratio-transmission (MRT) precoders in the downlink)\footnote{The deterministic equivalents are deterministic tight approximations of functionals of random matrices of finite size.}. This analysis was extended in~\cite{Papazafeiropoulos2015a,Papazafeiropoulos2014WCNC} by deriving DEs for the more sophisticated minimum mean-square error (MMSE) receivers (for the uplink) and regularized zero-forcing precoders (for the downlink). It is worthwhile to mention that in the analysis of this paper, we consider both the pilot contamination and the relative user mobility.

\subsection{Motivation-Central Idea}
This work relies on the observation that the promising HCNs, meant to be applied in 5G systems, have been mostly evaluated under the assumptions of perfect hardware and static environment, while these are highly unrealistic. In addition, no imperfect CSIT has been assumed regarding HCNs MIMO systems. Only in~\cite{Bai2016}, the authors have studied the uplink of massive MIMO systems with pilot contamination.  Hence, the fundamental question behind this work is ``how the transceiver impairments, the pilot contamination, and the channel aging affect the downlink performance of realistic HCNs, when imperfect CSIT is accounted?'' Motivated by recent performance analysis results, we are going to establish the theoretical framework modelling the additive residual transceiver hardware impairments (RATHIs) and imperfect CSIT, and identify the realistic potentials of HCNs before their final implementation. 

The main contributions are summarisedd as follows.
\begin{itemize}
  \item Contrary to existing work~\cite{Papazafeiropoulos2016b} which has studied the effect of RATHIs in the case of perfect CSIT,  we calculate the estimated channel due to RATHIs, pilot contamination,  and channel aging, and then, we  evaluate the effect of RATHIs and imperfect CSIT on the  performance of downlink realistic HCNs in terms of the coverage probability and the user rate. For the sake of comparison, we also present the results corresponding to perfect hardware and no user mobility.  As far as the authors are aware, this is the most general result in the literature that accounts for practical impairments and imperfect CSIT.
  \item Contrary to \cite{Bai2016}, where only the uplink performance was investigated for massive MIMO, we focus on the downlink. Moreover, we present more general results, since the proposed metrics  can describe scenarios with not only finite number of BS antennas, but also with large number of antennas. On the other hand, works such as \cite{Bai2016} describe only the large number of antennas regime.
 \item We investigate and elaborate on the impact of RATHIs and imperfect CSIT on the downlink coverage probability and the downlink achievable user rate. Actually, without being vague regarding the sources of imperfect CSIT in HetNets such as in \cite{Li2015a,Geraci2014,Li2015,Yang2016,Dhungana2017}, we consider the presence of pilot contamination, channel aging, and hardware impairments. Specifically, we show that the uplink hardware impairments have an equal effect on the estimation of the channel. However, the downlink hardware impairments behave differently. The BS transmit impairments degrade more the system performance than the user receive impairments. Moreover, the higher the time variation of the channel, the higher the degradation of the system. Accordingly, a proper system design should take into account these observations. 
 \item We focus on the design of a realistic HCN and investigate, if it is beneficial to transmit with space division multiple access (SDMA) or to employ fewer antennas under practical conditions.  In fact, single-stream transmission looks better than SDMA. In other words, the claim having more antennas is always beneficial is not necessarily correct, as it heavily depends on how the transmit antennas are used and which transmission/reception scheme is employed from a system perspective.
 \end{itemize}

    The   remainder of this paper is structured as follows.  Section~\ref{System} presents the basic parameters of the system model of a HCN with randomly located BSs having multiple antennas and serving multiple users. In Section~\ref{Presentation}, a description of the RATHIs is provided. In Section~\ref{estimation}, the channel estimation phase, considering pilot contamination and channel aging, is modelled  under the presence of RATHIs. Next, Section~\ref{downlink} exposes the   downlink transmission under RATHIs and imperfect CSIT.  Subsection~\ref{coverage} includes the  derivation and investigation of the coverage probability in a realistic HCN with multiple antenna  BSs impaired by RATHIs, when multiple users are served. Remarkably, a relative movement of the users with comparison to the BS antennas is considered as well. In Subsection~\ref{AverageAchievableRate1}, the derivation of the achievable user rate is provided under the same realistic conditions.  The numerical results are placed in Section~\ref{Numerical}, while Section~\ref{Conclusion} summariseds the paper.

\textit{Notation:} Vectors and matrices are denoted by boldface lower and upper case symbols. $(\cdot)^\T$, $(\cdot)^*$,  $(\cdot)^\H$, and $\tr\!\left( {\cdot} \right)$ express the transpose, conjugate, Hermitian  transpose, and trace operators, respectively. The expectation  and variance operators are denoted by $\EE\left[\cdot\right]$ and $\mathrm{Var}\left[\cdot\right]$. The $\mathrm{diag}\{\cdot\}$ operator generates a diagonal matrix from a given vector, and the symbol $\triangleq$ declares definition. The notations $\mathcal{C}^{M \times 1}$ and $\mathcal{C}^{M\times N}$ refer to complex $M$-dimensional vectors and  $M\times N$ matrices, respectively. The indicator function $ \mathds{1}(e)$ is $1$ when event $e$ holds and $0$ otherwise, and $\mathrm{J}_{0}(\cdot)$ is the zeroth-order Bessel function of the first kind.  Moreover,   $\mathrm{B}\left(x,y \right)$ is the Beta function defined in~\cite[Eq.~(8.380.1)]{Gradshteyn2007}, and $\Gamma\left( x,y \right)$ denotes the Gamma distribution with shape and scale parameters $x$ and $y$, respectively. Furthermore, $\underset{x \in A}{\cup}$ denotes the union with $A$ being an index set. Also, $\mathcal{L}_{I}\!\left(s \right)$ expressed the Laplace transform of $I$. Finally, $\bb \sim \cC\cN{(\b0,\mathbf{\Sigma})}$ represents a circularly symmetric complex Gaussian variable with zero-mean and covariance matrix $\mathbf{\Sigma}$.
 \section{System Model}\label{System} 
In this paper, we consider a network cellular MU system with a BS per cell that has multiple antennas and serves multiple users. The locations of the BSs are drawn according to an independent PPP $\Phi_{B}$ with density $\lambda_{B}$. In other words, we refer to the formulation of MU-MIMO HCNs.   Herein, we note that the locations of the users are also  modelled by an independent PPP $\Phi_{k}$   $\lambda_{k}$. In addition, the  BS in the $l$th cell, deployed with a number of BS antennas $M_{l}$, communicates with its associated users whose number is $K_{l}$ such that $K_{l}\le M_{l}$. Hence,  many degrees of freedoms are shared\footnote{Given this model, it is realistic that the downlink transmit power $p_{l}$,  the number of BS antennas $M_{l}$, and the number of users served by each BS $K_{l}$ differ across cells. However, for the sake of simplicity, we assume universal parameters in this setting, i.e., $p_{l}=p$, the number of BS antennas $M_{l}=M$, and the number of users served by each BS $K=K_{l}\le M_{l}=M $.}. Actually, we assume that each cell is large enough, i.e., the $l$th macrocell can  accommodate $K$ users connected with the nearest BS. In other words, $K$ users, that are independently distributed, belong to the Voronoi cell of this BS, while a Voronoi tessellation is structured by  the set of all these cells. Exploiting Slivnyak's theorem we are able to conduct the analysis by focusing on a typical user found at the origin~\cite{Chiu2013a}.  Basically, we focus on the downlink scenario of communication between the  BS and the associated users. During the uplink, channel estimation takes place, while we elaborate further on the downlink data transmission, where we derive the coverage probability and study its realistic behaviour.

Regarding HCNs, their evolution, started from the downlink single-input single-output (SISO) systems~\cite{Andrews2011}, has enabled the coexistence of multiple antenna strategies~\cite{Kountouris2012,Dhillon2013} such as beamforming and SDMA. In fact, HCNs and multiple antenna strategies  coexist and complement each other. Thus, they should not be studied in isolation, as happened in the premature literature in this area. For this reason, we focus on the impact of RATHIs on multiple antenna HCNs. 

All point-to-point channels are characterised by independent and identically distributed (i.i.d.) Rayleigh block fading models with unit mean. Moreover, the same time-frequency resources are shared by the users across all cells. Note that we assume that the channel coherence time is $T_{\mathrm{c}}$. At the same time, aiming to provide a realistic analysis, we assume imperfect CSIT due to pilot contamination, channel aging, and RATHIs. In other words, the BS is aware of the estimated channel, obtained during the training phase and  having duration $\tau$ symbols, while the downlink transmission phase has a duration of $T_{\mathrm{c}}-\tau$ symbols.

\section{Presentation of RATHIs}\label{Presentation} 
In practice, both the users and the BSs are affected by certain  additive  impairments~\cite{Schenk2008,Studer2010}. Although mitigation schemes are implemented in both the transmitter and receiver, these are not perfect. Therefore, RATHIs  still emerge by means of residual additive distortion noises~\cite{Schenk2008,Studer2010}. More concretely, at the transmitter side, an impairment emerges that  causes a mismatch between the intended signal and what is actually transmitted during the transmit processing, while at the receiver side the received signal appears a distortion.  

The study of the impact of RATHIs has originated from conventional wireless systems and has continued to 5G networks such as massive MIMO systems~\cite{Schenk2008,Goransson2008,Bjornson2012Optimal,Studer2010,Bjoernson2013,Zhang2014,Bjornson2014,Bjornson2015,Papazafeiropoulos2015b,PapazafeiropoulosMay2016}. Unfortunately, the majority of HCNs literature, except \cite{Papazafeiropoulos2016b} relies on the assumption of perfect transceiver hardware, although hardware imperfections exist. Reasonably, it is conjectured that by following the same path will increase the gap between theory and practice. 

Mathematically speaking, given the channel realisations, the conditional transmitter and  receiver  distortion noises for the $i$th link are modelled as Gaussian distributed, where their average power is proportional to the average signal power, as shown by measurement results~\cite{Studer2010}. 

Let us denote the transmit and receive nodes as nodes  $i$ and $j$, where $i=\mathrm{UE}$ and  $j=\mathrm{BS}$ for the uplink, while $i=\mathrm{BS}$ and  $j=\mathrm{UE}$ for the downlink.   In addtion, let $T_\mathrm{i}$ be the number of transmit antennas, i.e., $T_{\mathrm{UE}}=1$ for the uplink and $T_{\mathrm{BS}}=M$ for the downlink. $\bQ_{i}[n]$ is the transmit covariance matrix at time instance $n$ of the corresponding node with diagonal elements $q_{\mathrm{i}_1}[n],\ldots,q_{T_\mathrm{i}}[n]$, e.g., if the transmitter node is the UE, $\bQ_{i}[n]$ degenerates to a scalar $Q_{\mathrm{UE}}[n]$.  Specifically, the RATHIs at the transmitter and the receiver are given by
\begin{align}
 \etv_{\mathrm{t},n}^{\mathrm{i}}&\sim \cC\cN\left( \b0,\bm \Lambda^{\mathrm{i}}_{n} \right)\label{eta_t} \\
 \etv_{\mathrm{r},n}^\mathrm{j}&\sim \cC\cN \left( \b0,\bm \Upsilon^{\mathrm{j}}_{n} \right)\label{eta_r},
\end{align}
where $\bm \Lambda^{\mathrm{i}}_{n}= \kappa_{\mathrm{t}_\mathrm{i}}^{2}\mathrm{diag}\left( q_{1}[n],\ldots,q_{T_\mathrm{i}}[n] \right)$ and 
$\bm \Upsilon^{\mathrm{j}}_{n} =\kappa_{\mathrm{r}_{\mathrm{j}}}^{2}\|x_{k,n}\|^{-\alpha}\sum_{k=1}^{|j |} \bh_{k}^{\H}[n]\bQ_{k}[n]\bh_{k}[n] $. Note that if $j=\mathrm{UE}$, then $|j|=1$. The circularly-symmetric complex Gaussianity can be justified by the aggregate contribution of many impairments\footnote{The additive distortions   are time-dependent because they take new realisations for each new data signal.}. The proportionality parameters $\kappa_{\mathrm{t}_i}^{2}$ and $\kappa_{\mathrm{r}_j}^{2}$   describe the severity of the residual impairments at the transmitter and the receiver side. In applications, these parameters are met as the error vector magnitudes (EVM) at each transceiver side~\cite{Holma2011}. The procedure for obtaining the knowledge of the estimated  channel is described in the following section.

\section{Channel Estimation}\label{estimation} 
 The transmit signal by each user to its BS is attenuated with distance $r$ as $r^{-\alpha}$, where $\alpha$ is the path-loss exponent parameter. During the channel estimation, an effect, known as pilot contamination due to the  re-use of the pilot sequences, might arise. It is worthwhile to mention that both the system model and the proposed expressions below can describe any number of BS antennas, e.g., both small and large number of antennas. Indeed, the results are quite general, and will be derived by means of a common analysis followed in the study of HetNets. In other words, contrary to works for massive MIMO that can describe only the large antenna regime, the proposed expressions can describe the whole spectrum in terms of the number of antennas (from small to large number). Notably, the pilot contamination concerns systems with any number of antennas. However, it is negligible for small number of antennas, while its presence is observable for large MIMO cellular systems \cite{Marzetta2010}. Hence, since our model is able to describe any number of antennas (even a large number), pilot contamination is meaningful in the current work.
 \subsection{Pilot Contamination}
 Each BS estimates the CSI during an uplink training phase, where the sharing of the same band of frequencies leads to degradation of the performance of the system due to pilot contamination. If the subscript $\mathrm{tr}$ denotes the training stage, the noisy observation  of the channel vector from the  user $k$ (typical user at the origin) at time instance $n$, transmitting each pilot symbol with average power $\rho^{\mathrm{UE}}_{\mathrm{up}}$ $\left( p_{\mathrm{tr}}=\tau \rho^{\mathrm{UE}}_{\mathrm{up}} \right)$ to its associated BS in the presence of RATHIs,  is given by
\begin{align}
\!\!\!\tilde{\by}_{k,\mathrm{tr}}&[n]
\!= \! \underbrace{\bh_{k}[n]\| x_{k,n}\|^{-\alpha/2}}_{\minibox[c]{Desired \\signal}}\!+\!\underbrace{\bh_{k}[n]\etv_{\mathrm{t},n}^{\mathrm{UE}}\| x_{k,n}\|^{-\alpha/2}}_{\minibox[c]{UE transmit\\ impairment}} \nn\\&\!+\!\!\!\!\!\underbrace{\etv_{\mathrm{R},n}^{\mathrm{BS}}}_{\minibox[c]{BS receive\\ impairment}}
\!\!+\!\! \!\underbrace{\sum_{l \in \Phi_{B}/ b_{0}}\!\!\! \!\!\!\bg_{lk}[n]\|y_{k,n}\|^{-\alpha/2}}_{\minibox[c]{Interference\\part}}\! \!+\! \!\underbrace{\frac{1}{\sqrt{p_{\mathrm{tr}}}} \bN_{\mathrm{tr}}[n]\bpsi^{\H}_{k}}_{\mbox{noise}},\label{eq:Ypt3}
\end{align}
where $\bh_{k}[n]\in \mathcal{C}^{M\times 1}$ is the desired channel vector from the $k$th user in the current cell (located at the origin) to its associated BS located at $x_{k,n}$, while the interference term $\bg_{lk}[n]\in \mathcal{C}^{M\times 1}$ is the channel vector corresponding to the link from the $k$th user of the $l$th cell located at $y_{lk,n}\in \mathcal{R}^{2}$. The vector $\bpsi_{k}\in \mathcal{C}^{\tau\times 1}$ denotes the training sequence of the $k$th user with $\bpsi_{k}\bpsi_{k}^{\H}$=1 and  $\bN_{\mathrm{tr}}[n]\in \mathcal{C}^{M\times \tau}$ is a spatially white additive Gaussian
noise matrix with i.i.d entries at the current base station  during the training stage, which are  distributed as $\mathcal{CN}\left( 0,1 \right)$. Similarly, the channel vectors $\bg_{lk,n} \forall l,n$ are Gaussian distributed as $\mathcal{CN}\left( \b0,\mathbf{ I}_{M} \right)$, and they are independent across cells and user distances. Based on~\eqref{eta_t} and~\eqref{eta_r} for $\mathrm{i}=\mathrm{UE}$ and $\mathrm{j}=\mathrm{BS}$, respectively, and after making the appropriate substitutions, the hardware impairments at the uplink are written such that their variances are given by
\begin{align}
\Lambda^{\mathrm{UE}}_{n}&=\kappa_{\mathrm{t}_\mathrm{UE}}^{2}\rho_{up}^{\mathrm{UE}}\nn\\
\bm \Upsilon^{\mathrm{BS}}_{n}&=\kappa_{\mathrm{r}_\mathrm{BS}}^{2} \rho_{up}^{\mathrm{UE}}\|x_{k,n}\|^{-{\alpha}}\sum_{k=1}^{K}  \mathrm{diag}\left( |h_{k}^{1}[n]|^{2},\ldots,|h_{k}^{M}[n]|^{2} \right)\nn,
\end{align}
where $\bQ_{\mathrm{UE}}[n]$ is now deterministic and $\rho_{up}^{\mathrm{UE}}=\tr\left( \bQ_{\mathrm{UE}}[n] \right)$


After applying the minimum mean square error (MMSE) estimation to~\eqref{eq:Ypt3}, in order to estimate $\bh_{k}[n]$ conditioned on the distance $\| x_{k,n}\|$ from the $k$th user, the current BS  obtains
\begin{align}\label{estimated1} 
& \hat{\bh}_{k}[n]\Big|_{\| x_{k,n}\|}\!=\!{\mathrm{E}\!\left[\bh_{k}[n]\tilde{\by}_{k,\mathrm{tr}}^{\H}[n]  \right]}{\mathrm{E}^{-1}\!\left[\tilde{\by}_{k,\mathrm{tr}}[n]\tilde{\by}_{k,\mathrm{tr}}^{\H}[n]\right]}\tilde{\by}_{k,\mathrm{tr}}[n]\nn\\
 &=\frac{{\sqrt{p_{\mathrm{tr}}}}\| x_{k,n}\|^{-\alpha/2}}{p_{\mathrm{tr}}\| x_{k,n}\|^{-\alpha}\left( \kappa^{2}_{\mathrm{r},\mathrm{BS}}+\kappa^{2}_{\mathrm{t},\mathrm{UE}} \right)+p_{\mathrm{tr}}\mathrm{Var}\left( I_{n} \right)+1}\tilde{\by}_{k,\mathrm{tr}}[n],\!\!
\end{align}
where $\displaystyle I_{n}\!\!=\!\!\sum_{l \in \Phi_{B}/ b_{0}}\!\! \!\! \bg_{lk}[n]\|y_{k,n}\|^{-\alpha/2}$ represents the inteference contribution at time $n$ and $\mathrm{Var}\left( \cdot \right)$ is the variance operator conditioned on $\| x_{k,n}\|$. Specifically, conditioning on $\| x_{k,n}\|$, $\mathrm{Var}\left( I_{n} \right)$ can be derived by using Campbell's theorem~\cite{Chiu2013} as
\begin{align}
 \mathrm{Var}\!\left( I_{n} \right)&=2 \EE\!\left(\|\bg_{lk}[n]\|^{2}\right)\pi \lambda_{B}\int_{\| x_{k,n}\|}^{\infty}r^{1-\al}\mathrm{d}r\nn\\
 &=2 M \pi \lambda_{B}\frac{\| x_{k}\|^{2-\alpha}}{\alpha-2},
\end{align}
where we have used that the square of $\bg_{lk}[n]$ follows a $\Gamma\left(M,1\right)$ distribution with mean value equal to $M$.  Clearly, increase of the path-loss exponent and BSs' density, increases the variance of the interference.

As a result, the estimated channel  $\hat{\bh}_{k}$ and estimation error vectors $\bee_{k}[n]={\bh}_{k}[n]-\hat{\bh}_{k}[n]$ at time instance $n$ are uncorrelated and distributed as $\hat{\bh}_{k}[n]\sim \mathcal{CN}\left( \b 0, \sigma_{\hat{\bh}_{k}}^{2}  \mathbf{I}_{M}\right)$ and $\hat{\bee}_{k}[n]\sim \mathcal{CN}\left( \b 0, \sigma_{\hat{\bee}_{k}}^{2}  \mathbf{I}_{M}\right)$ with
\begin{align}
 \sigma_{\hat{\bh}_{k}}^{2}=\frac{\left( \al-2 \right)}{ 
 \left( \kappa^{2}_{\mathrm{r},\mathrm{BS}}+\kappa^{2}_{\mathrm{t},\mathrm{UE}}+1 \right)\left( \al-2 \right)+  2 M \pi \lambda_{B}\| x_{k}\|^2 +\| x_{k}\|^\al}
    \end{align}
and 
\begin{align}
\sigma_{\hat{\bee}_{k}}^{2}\!=\!\frac{ \left(\kappa^{2}_{\mathrm{r},\mathrm{BS}}\!+\!\kappa^{2}_{\mathrm{t},\mathrm{UE}}\right)\left( \al\!-\!2 \right)\!+ \!2 M \pi \lambda_{B}\| x_{k}\|^2 }{\left(\kappa^{2}_{\mathrm{r},\mathrm{BS}}\!+\!\kappa^{2}_{\mathrm{t},\mathrm{UE}}\!+\!1 \right)\!\left( \al\!-\!2 \right)\!+\! 2 M \pi \lambda_{B}\| x_{k}\|^2+\| x_{k}\|^\al }.
\end{align}

\begin{remark}
If we set the distortion parameters in~\eqref{estimated1} or in any other expression including the RATHIs equal to zero, we result in the expressions corresponding to the ideal uplink model, which does not account for RATHIs.  
\end{remark}
\begin{remark}
 Interestingly, in the interference-limited regime, the covariances of the estimated channel and the estimation error do not depend on the uplink transmit power $p_{\mathrm{tr}}.$ Moreover, in these expressions and under these conditions, the distance from the typical user is not raised to the path-loss parameter $\al$. Thus, the covariances are not environment-dependent.
\end{remark}
\begin{remark}
The uplink RATHIs, $\kappa^{2}_{\mathrm{t},\mathrm{UE}}$ and $\kappa^{2}_{\mathrm{r},\mathrm{BS}}$, have the same effect on the accuracy of the estimated channel. Therefore, equal caution should be given in the quality of the user transmit and BS receive hardware.
\end{remark}

\subsection{Channel Aging}
Another additive reason with negative effect that necessitates the estimation of the channel  is imposed by the  relative movement of the users with comparison to the BS antennas. In such case, the description of the channel can be made by a time-varying model. In mathematical terms, the relationship of the current sample of the  channel with its past samples, i.e., the time varying-model at time-slot $n$, is represented by an autoregressive-model of order $1$ that includes  the second order statistics of the channel in terms of its autocorrelation function being  generally a function of velocity of the user, the  propagation geometry, and the antenna characteristics. Basically, this is a Gauss-Markov model of low order $\left( 1 \right)$ for reasons of computational complexity and tractability.  Actually, the current channel between the BS  and the typical user belonging to its cell is modelled as
\begin{align}
\bh_{k}[n]  =& \delta \bh_{k}[n-1] + \bee_{k}[n],\label{eq:GaussMarkoModel}
\end{align}
where $\bh_{k}[n-1]$ is the channel in the previous symbol duration and $\bee_{k}[n] \in \bbC^{N}$ is an uncorrelated channel error due to the channel variation modelled as a stationary Gaussian random process with i.i.d.~entries and distribution $\cC\cN(\b0,(1-\delta^2)\Id_{N}$~\cite{Vu2007}.

Given that the variation of the channel is described by means of its second order statistics, we employ the autocorrelation function of the channel, which is an appropriate measure. We choose the Jakes model for representing the autocorrelation function, which is widely accepted due to its generality and simplicity \cite{Baddour2005}. The Jakes model describes a propagation medium with two-dimensional isotropic scattering and a monopole antenna at the receiver \cite{WCJr1974}. Mathematically, the normalised discrete-time autocorrelation function of the fading channel is expressed by
\begin{align}
r_{h}[k] 
=& \mathrm{J}_{0}(2 \pi f_{D}T_{s}|k|).\label{eq:scalarACF}
\end{align}
Specifically,  $f_{D}$ and $T_{s}$ are the maximum Doppler shift and the channel sampling period, respectively. Regarding the maximum Doppler shift $f_{D}$, it can be expressed  in terms of the relative velocity of the  $v$, i.e., $f_{D}=\frac{v f_{c}}{c}$, where $c=3\times10^{8}~\nicefrac{m}{s}$ is the speed of light and $f_{c}$ is the carrier frequency.  Note that we assume that the base stations have perfect knowledge of $\delta=r_{h}[1]$.

Interestingly, we are able to combine both effects of pilot contamination and time-variation of the channel according to~\cite{Truong2013,Papazafeiropoulos2015a}. Specifically, we have that the fading channel at time slot $n$ can be expressed by
\begin{align}   
 \bh_{k}[n]&=\delta {\bh}_{k}[n-1]+\bee_{k}[n]\nonumber\\
&=\delta \hat{\bh}_{k}[n-1]+ \tilde{\bee}_{k}[n],\label{eq:MMSEchannelEstimate}
\end{align}
where $\hat{\bh}_{k}[n-1]$ and $\tilde{\bee}_{k}[n]= \delta \tilde{\bh}_{k}[n-1]+\bee_{k}[n]\sim \mathcal{CN}\left(\b0,\sigma_{\tilde{\bee}_{k}}^{2} {\mathrm{ \bI}}_{M} \right)$ with $\sigma_{\hat{\bee}_{k}}^{2}=\left(1-\delta^{2}\sigma_{\hat{\bh}_{k}}^{2} \right)$
are mutually independent. In other words, the estimated channel at time $n$ is now $\hat{\bh}_{k}[n]=\delta \hat{\bh}_{k}[n-1]$.
\begin{remark}
Imperfect CSIT is the result of three different sources, namely, i) pilot contamination, ii) channel aging, and iii) RATHIs. 
\end{remark}

\section{Downlink Transmission}\label{downlink} 
The channel power distribution of a link depends upon its physical representation, i.e., if it describes the desired or the interference contribution to the received signal or we address to a multi-antenna BS deployment with single or multi-user transmission.  However, in order to arrive at the stage of the statistical distribution of the power, we have to model the downlink transmission. Note that if $\delta=0$, we obtain a static environment with no user mobility.

This paper assumes linear precoding by the matrix $\bW[n]=[\bw_1,\ldots,\bw_K]\in \mathbb{C}^{M\times K}$, employed by the associated BS, which multiplies the data signal vector $\bd[n] = \big[d_{1}[n],\dots,~d_{K}[n]\big]^\T \in \mathbb{C}^{M}\sim \mathcal{CN}(\b0,\bI_{K})$ for all users in that cell. In our case, we employ ZF precoding that has the form 
\begin{align}\label{precoder}
\hat{\bW}[n]&=\bar{\bH}[n]\left( \bar{\bH}^{\H}[n]\bar{\bH}[n] \right)^{-1}\\
&=\delta^{-1} \bar{\bH}^{\dagger}[n-1]=\delta^{-1} \hat{\bW}[n-1],  \label{delayedPrec}                                                                                                                                                                                                                                                                                                                                                                            \end{align}
where $\bar{\bH}[n]$ is the normalised version of $\hat{\bH}[n]$ and is related to $\hat{\bH}[n]$ as defined in Appendix~\ref{SINRproof}, while $\left( \cdot \right)$ denotes the pseudo-inverse of a matrix. Thus, the precoder is normalised and  the average transmit power per user of the associated BS is constrained to $p$, i.e., $\mathbb{E}\big[\mathrm{tr}\left( \hat{\bW}[n]\hat{\bW}^{\H}[n] \right)\big]=1$. 
Note that $\hat{\bW}[n]$ denotes the ZF precoder of the associated BS at time $n$ as well as  in~\eqref{delayedPrec} we have introduced the user mobility effect for the $k$th user by means of $\hat{\bh}_{k}[n]=\delta \hat{\bh}_{k}[n-1]$. 

Thus, the received signal from the associated BS to the typical  user at $x_{k,n}$  during the $n$th time-slot after applying the ZF precoder, accounting for a quasi-static block fading model with frequency-flat fading channels varying for symbol to symbol and RATHIs, can be expressed as\footnote{Since we focus on the investigation of the interference-limited SIR, its expression does not depend on the transmit power.}
\begin{align}
 y_{k}[n]&=\bh_{k}^{\H}[n]|x_{k,n}\|^{-\al/2}\left( \bs[n]+\etv^{\mathrm{BS}}_{\mathrm{t},n}  \right)+\eta^{\mathrm{UE}}_{\mathrm{r},n}\nn\\
 &+\!\!\!\sum_{l\in \Phi_{B}/x_{k,n}}\!\!\!\!\bg_{lk}^{\H}[n] \bs_{lk}[n]\|y_{lk,n} \|^{-\al/2}\label{signal} ,
\end{align}
where  $\bs[n] =\bW[n] \bd[n] \in \mathbb{C}^{M \times 1}$ is the  transmit signal vector for the $k$th user with covariance matrix $\bQ_{\mathrm{BS}}[n]=\EE\left[ \bs[n]\bs^{\H}[n]\right] $ and   $p=\tr\left( \bQ_{\mathrm{BS}}[n] \right)$ is the associated average transmit power.  The channel vector $\bh_{k}[n] \in \mathbb{C}^{M\times 1}$  denotes the desired  channel vector between the associated BS located at $x_{k,n}\in \mathbb{R}^{2}$ and the typical user at  time-instance $n$. Similarly,  $\bg_{lk}[n] \in \mathbb{C}^{M\times 1}$ expresses the interference channel vector from the BSs found at $y_{lk,n}\in \mathbb{R}^{2}$ far from the typical user at time-instance $n$.  Especially, in the case of Rayleigh fading, the channel power distributions of both the direct and the interfering links follow the Gamma distribution~\cite{Huang2011}.  

Given that we have assumed the realistic scenario of  RATHIs as well as imperfect CSI due to pilot contamination and time-variation of the channel (channel aging) (see~\eqref{eq:MMSEchannelEstimate} and~\eqref{delayedPrec}),  the received signal by user $k$ can be written as\footnote{ We assume that the RATHIs from other BSs are negligible due to the increased path-loss, while   the transmit hardware impairment depends only on the transmit signal power from the tagged BS and not from the path-loss.},\footnote{Herein, we assume that the thermal noise is negligible as compared to the distortion noises and interference from the other cells as shown by simulations and previously known results. Hence, its omission is reasonable. However, it can be included in the proposed analysis with no extra special effort.}
\begin{align}
 y_{k}[n]&= \hat{\bh}_{k}^{\H}[n-1]\hat{\bw}_{k}[n-1] \bd[n] \|x_{k,n} \|^{-\al/2}\nn\\
 &+\delta^{-1}\!\tilde{\bee}_{k}^{\H}[n]\hat{\bW}[n-1] \bd[n]\|x_{k,n} \|^{-\al/2}\nn\\
&+ \bh_{k}^{\H}[n] \|x_{k,n} \|^{-\al/2}\etv^{\mathrm{BS}}_{\mathrm{t},n}+\eta^{\mathrm{UE}}_{\mathrm{r},n}\nn\\
  &+\!\sum_{l\in \Phi_{B}/x_{k,n}}\!{\bg}_{lk}^{\H}[n] \hat{{\bW}}_{l }[n]\bd_{l}[n]\|y \|^{-\al/2}, \label{filteredsignal}
\end{align}
where we have used~\eqref{eq:MMSEchannelEstimate} for replacing the current desired channel\footnote{The replacement concerns only the current desired channel because the interference part is not of direct interest and can be kept in the initial form.}. In addition, we have used~\eqref{delayedPrec} to  express the current precoder in terms of its delayed instance known at the BS.  The first term in~\eqref{filteredsignal} represents the desired signal contribution in the typical current cell, while the second term describes the estimation error effect. Further, the third   term describes the contribution due to the transmit BS impairment. The fourth term represents the receive impairment part because of the hardware impairments at the user side. Moreover, the next term characterises the  inter-cell interference coming from BSs in different cells. 

$\etv_{\mathrm{t},n}^{\mathrm{BS}}\sim \cC\cN\left( \b0,\bm \Lambda^{\mathrm{BS}}_{n} \right)$ and $\eta_{\mathrm{r},n}^{\mathrm{UE}}\sim \cC\cN \left( \b0, \Upsilon^{\mathrm{UE}}_{n} \right)$ are the residual downlink additive Gaussian distortions at the BS and the UE, which are given by~\eqref{eta_t} and~\eqref{eta_r} for $\mathrm{i}=\mathrm{BS}$ and $\mathrm{j}=\mathrm{UE}$, respectively. Specifically, we have
\begin{align}
 \bm \Lambda^{\mathrm{BS}}_{n}&=\kappa_{\mathrm{t}_\mathrm{BS}}^{2} \mathrm{diag}\left( q_{1}[n],\ldots,q_{M}[n] \right)\\
  \Upsilon^{\mathrm{UE}}_{n}&=\kappa_{\mathrm{r}_\mathrm{UE}}^{2}\|x_{k,n}\|^{-{\al}}\bh_{k}^{\H}[n]\bQ_{\mathrm{BS}}[n]  \bh_{k}[n].
\end{align}

\begin{remark}
The receive distortion at user $k$ includes the path-loss coming from the associated BS. 
\end{remark}

The achievable downlink SIR $ \gamma_{k}$ from the  associated BS to the $k$th user  can be defined as in~\eqref{eq: general sum_rate} at the top of the next page. Note that we have assumed encoding of the message over many realisations of all sources of randomness in the model including noise, channel estimate error, and RATHIs. Based on this remark, the hardware impairments are written as
$  \etv_{\mathrm{t},n}^{\mathrm{BS}}\sim \cC\cN(\b0,{\kappa_{\mathrm{t}_{\mathrm{BS}}}^{2}}  \Id_{K})$  and  $\eta_{\mathrm{r},n}^{\mathrm{UE}}\sim \cC\cN(\b0,\kappa_{\mathrm{r}_{\mathrm{UE}}}^{2}\|x_{k,n}\|^{-\al}\|\bh_{k}[n]\|^{2} )$.

\begin{remark}
If  the distortion parameters in~\eqref{filteredsignal} or in any other expression including the downlink RATHIs are set to zero, we result in the expressions corresponding to the ideal downlink model, which does not account for RATHIs~\cite{Dhillon2013}.  
\end{remark}

In order to develop this general model, we assume that the desired channel power from the associated BS at time $n=\tau+1,\ldots, T_{\mathrm{c}}$ located at $x_{k,n}\in \mathbb{R}^{2}$ to the typical user $k$, found at the origin, is given by $Z_{k}[n]$, while the interfering link from other BSs, i.e., the inter-cell interference from other cells' BSs (located at $y_{lk,n} \in \mathrm{R}^{2}$) is denoted by $I_{l}[n]$. 
\setcounter{eqnback}{\value{equation}} \setcounter{equation}{18}
\begin{figure*}[!t]
\begin{align}\label{eq: general sum_rate}
     \gamma_{k}
    &=
        \frac{
             Z_{k}[n] \|x_{k,n} \|^{-\al}
            }{\left( E_k[n]+
            I_{ \etv_{\mathrm{t}}}[n]
            +I_{ \etv_{\mathrm{r}}}[n]\right) \|x_{k,n} \|^{-\al}+
     I_{ l}[n] }.
\end{align}
\hrulefill
\end{figure*}
\setcounter{eqnback}{\value{equation}} \setcounter{equation}{19}

\begin{proposition}\label{SINR} 
 The  SIR of the downlink transmission from the associated BS to  the typical user, accounting for   RATHIs and imperfect CSI due to pilot contamination and time variation of the channel due to user mobility,  can be represented  as in~\eqref{eq: general sum_rate}
 where the probability density function (PDF) of the desired signal power is obtained by
 \begin{align}\label{eq PDF1 1}
    p_{Z_{k}[n]}\left( z \right)=
       \frac{
            e^{-z/\sigma_{\hat{\bg}_{k}}^{2}}
            }{
            \left(M-K\right)!
         \sigma_{\hat{\bh}_{k}}^{2}
            }
        \left(
            \frac{
                z
                }{
               \sigma_{\hat{\bh}_{k}}^{2}
                }
        \right)^{M-K}, ~ z \geq 0    
        \end{align}
        and the various terms are given by
 \begin{align}
 Z_{k}[n]&=|\hat{\bh}_{k}^{\H}[n-1]\hat{\bw}_{k}[n-1]|^2 \\
 E_k[n]&=\delta^{-2} \left(1+\kappa_{\mathrm{t}_{\mathrm{BS}}} ^{2}\right)
               \big\|
                    \tilde{\bee}_{k}^{\H}[n]\hat{\bW}[n-1]
              \big\|^2\\
          I_{ \etv_{\mathrm{t}}}[n]&= \kappa_{\mathrm{t}_{\mathrm{BS}}} ^{2}|\hat{\bh}_{k}^{\H}[n-1]\hat{\bw}_{k}[n-1]|^2\\
I_{ \etv_{\mathrm{r}}}[n]&={\kappa_{\mathrm{r}_{\mathrm{UE}}}^{2}}\|{\bh}_{k}[n]\|^{2}\\
I_{l}[n]&= 
            \sum_{l \in \Phi_{B}/x_{k,n} }   \big|{\bg}_{lk}[n]\hat{\bW}_{l}[n] \big|^{2}\|y \|^{-\al}\nn\\
           &= \sum_{l \in \Phi_{B}/x_{k,n}} g_{lk,n}\|y \|^{-\al}.
            \end{align}
\end{proposition}

Note that $\bw_{k}$ is the $k$th column of $\bW_{k}$.
In addition, we have $
I_{ \etv_{\mathrm{t}}}[n]\sim {p_{k}\kappa_{\mathrm{BS}}^{2}}\Gamma (\Delta,\sigma_{\hat{\bh}_{k}}^{2})$,  and $
I_{ \etv_{\mathrm{r}}}[n]\sim  p_{k}{\kappa_{\mathrm{UE}}^{2}}\Gamma (M,1)$
while the total interference from all the other base stations found at a distance $\|y_{lk,n}\|$ from the typical user 
is
 $I_{l}[n] 
           = \sum_{l \in \Phi_{B}/x_{k,n}} p_{l}g_{lk}[n]\|y_{lk,n}\|^{-\al}$, 
            where $g_{lk}[n]=\big|{\bg}_{lk}[n]\hat{\bW}_{l}[n] \big|^{2}\sim \Gamma (K,1)$.

\begin{proof}
See Appendix~\ref{SINRproof}.
\end{proof}

It is worthwhile to mention that the distance $\| x_{k,n}\|$ in the uplink and the downlink expresses two different  variables and no confusion should arise. Specifically, during the uplink, the covariance of the estimated channel that includes $\| x_{k,n}\|^{-\alpha}$ is calculated for a given distance. Thus, it is deterministic, while its instance in the downlink is randomly distributed.

\section{Main Results}
Herein, we present the coverage probability and the achievable rate of the typical user, which are the main results of this work.

\subsection{Coverage Probability}\label{coverage} 
The focal point of this subsection is the derivation of an upper bound of the downlink coverage probability of the typical user in a multiple antenna HCN under the presence of RATHIs as well as imperfect CSIT due to pilot contamination and channel aging. Please note that the derivation of a lower bound, which is quite meaningful, is left for future work. Specifically, this subsection first presents a formal definition of the coverage probability with random BS locations drawn from a PPP. Next, the main result is provided with the technical derivation given in Appendix~\ref{CoverageProbabilityproof}. Remarkably, although we start from an abstract definition, we result in the most general and versatile expression known in the literature towards a more realistic assessment of a network with randomly located BSs.
Interestingly, we present a more general result than~\cite{Papazafeiropoulos2016b}, since now imperfect CSIT is assumed, when the BSs are randomly located. It is based on the calculation of the Laplace transforms provided by means of Proposition~\ref{LaplaceTransform} and Lemma~\ref{LaplaceTransformGamma} provided below. 
\begin{definition}[\!\!\cite{Dhillon2013,Papazafeiropoulos2016b}]
A typical user  is  in  coverage if its effective  downlink SIR from at least one of the randomly located BSs in the network is higher
  than the corresponding target $T$. In general, we have
 \begin{align}
 p_{c}\left(  T,\lambda_{B},\al,\delta,q \right)\triangleq\EE\left[\mathds{1}\!\left( 	\underset{x \in \Phi_{B}}{\cup} \mathrm{SIR\left( x \right)>T} \right)  \right],
\end{align}
where   $q$ defines a set of parameters. Specifically, we define $q\triangleq\{\kappa_{\mathrm{t}}^{\mathrm{UE}},\kappa_{\mathrm{r}}^{\mathrm{UE}},\kappa_{\mathrm{t}}^{\mathrm{BS}},\kappa_{\mathrm{r}}^{\mathrm{BS}}\}$. 
 \end{definition}

\begin{theorem}\label{theoremCoverageProbability} 
The upper bound of the downlink probability of coverage  $p_{c}\left(  T,\lambda_{B},\al,\delta,q  \right)$ in a general cellular network with randomly distributed multiple antenna BSs, accounting for RATHIs and imperfect CSIT due to pilot contamination and channel aging, is given by
 \begin{align}
\!\!&p_{c}\left(  T,\lambda_{B},\al,\delta,q \right)\!\le\!\!\lambda_{B}\int_{l \in  \mathds{R}^2}\!\!
\sum_{i=0}^{\Delta-1}\!\sum_{k=0}^{i}\!\sum_{u_1+u_2+u_3=i-k}\!\!\binom{i}{k}\!\nn\\
                  &\!\times\binom{i\!-\!k}{u_1+u_2+u_3}\!\frac{\left( -1 \right)^{i}\!s^{k}  }{i!}\frac{\mathrm{d}^{u_1}}{\mathrm{d}s^{u_1}}\mathcal{L}_{E_k}\left( l^{-\al}s  \right)\!  \nn\\
                  &\times  \frac{\mathrm{d}^{u_2}}{\mathrm{d}s^{u_2}}\mathcal{L}_{I_{ \etv_{\mathrm{t}}}}\!\!\left(l^{-\al} s \right)\frac{\mathrm{d}^{u_3}}{\mathrm{d}s^{u_3}}\mathcal{L}_{I_{ \etv_{\mathrm{t}}}}\!\!\left(l^{-\al} s \right)\!   \frac{\mathrm{d}^{k}}{\mathrm{d}s^{k}} \mathcal{L}_{I_{l}}\!\!\left(s \right)\mathrm{d}l,\label{coverageprobability} \!
\end{align}
where  $l=\|x\|$ and  $s=\frac{{T}}{\sigma_{\hat{\bg}_{k}}^{2}}l^{a}$,   while $\mathcal{L}_{I_{ \etv_{\mathrm{r}}}}\!\left(s \right)$,   $\mathcal{L}_{I_{ \etv_{\mathrm{t}}}}\!\left(s \right)$, and $\mathcal{L}_{I_{l}}\!\left(s \right)$ are the Laplace transforms of the powers of the receive distortion, transmit distortion, and  interference  coming from other BSs.
\end{theorem}
\begin{proof}
See Appendix~\ref{CoverageProbabilityproof}\footnote{Although the expressions [21]-[25] could be characterised in the asymptotic regime, this action would result in deterministic expressions that could not be manipulated statistically, in order to derive the coverage probability.}.
\end{proof}

\begin{proposition}\label{LaplaceTransform} 
The Laplace transform of the interference power of a  general cellular network with randomly distributed multiple antenna BSs having RATHIs and imperfect CSIT is given by
\begin{align}
&\mathcal{L}_{I_{l}}\!\left({s} \right)=\exp{\!\left( - {s}^{\frac{2}{a}} \mathcal{C}\left( \al, K\right) \right)},
\end{align}
where  $\mathcal{C}\left( \al, K\right)=\frac{2 \pi \lambda_{B}}{a} \sum_{m=0}^{K}\binom{K}{m}\mathrm{B}\left( K-m+\frac{2}{a},m-\frac{2}{a} \right)$.
\end{proposition}
\begin{proof}
See Appendix~\ref{LaplaceTransformproof}.
\end{proof}
\begin{lemma}\label{LaplaceTransformGamma} 
 The Laplace transforms of the parts, describing the RATHIs $I_{ \etv_{\mathrm{t}}}$ and $I_{ \etv_{\mathrm{r}}}$ as well as the estimation error, are given by
 \begin{align}
\mathcal{L}_{I_{ \etv_{t}}}\!\left(s \right)&=\frac{1}{\left( 1+\kappa_{\mathrm{t}_{\mathrm{BS}}}^{2} s \right)^{\Delta}}, \\
\mathcal{L}_{I_{ \etv_{r}}}\!\left(s \right)&=\frac{1}{\left( 1+\kappa_{\mathrm{r}_{\mathrm{UE}}}^{2} s \right)^{M}},\\
\mathcal{L}_{E_k}\!\left(s \right)&=\frac{1}{\left( 1+\delta^{-2} \left(1+\kappa_{\mathrm{t}_{\mathrm{BS}}} ^{2}\right)\sigma_{\hat{\bee}_{k}}^{2} s \right)^{\Delta}}.
 \end{align}
\end{lemma}
\begin{proof}
The first two  Laplace transforms are easily obtained, since $I_{ \etv_{\mathrm{t}}}$ and $I_{ \etv_{\mathrm{r}}}$ follow the scaled Gamma distributions with scaled parameters ${\kappa_{\mathrm{t}_{\mathrm{BS}}}^{2}}$ and ${\kappa_{\mathrm{r}_{\mathrm{UE}}}^{2}}$, as mentioned in Appendix~\ref{SINRproof}. In the same appendix, it is shown that the estimation error is a scaled Gamma distribution.
\end{proof}
\begin{remark}
The expressions corresponding to the Laplace transforms of the distortion noises reveal that  the downlink RATHIs show a different behaviour than the uplink impairments. Specifically, $\kappa^{2}_{\mathrm{t},\mathrm{BS}}$ has a greater impact than $\kappa^{2}_{\mathrm{r},\mathrm{UE}}$, since $M >\Delta$. As a result, the quality of the transmit BS impairments should be kept above a certain standard that would not allow significant distortion of the system performance. Moreover, the Laplace transform of the transmit impairments depends on both the number of BS antennas and users, while the Laplace transform of the receive impairments depends only on the number of BS antennas. In other words, the higher the number of users is, the smaller the impact of the transmit impairments becomes, since the corresponding Laplace transform increases. Furthermore, the Laplace transform of the estimation error depends on the channel aging by means of $\delta$ and both transmit and receive hardware impairments. In fact, the more severe the channel aging is by means of smaller $\delta$ (higher users' velocity), the smaller  $\mathcal{L}_{E_k}\!\left(s \right)$ becomes.
\end{remark}
\begin{lemma}
The PDF of    $g_{lk,n}$, describing the interfering marks, is $\Gamma (K,1)$ distributed.
\end{lemma}
\begin{proof}
Given that the precoding matrices $\hat{\bW_{l}}[n] $ have unit-norm columns and each instance does not depend on $\bg_{lk}[n]$ and $\bar{\bg}_{lk}[n]$ (the normalised version of $\bg_{lk}[n]$),  $\hat{\bW_{l}}[n] $  are independent isotropic
unit-norm random vectors. Hence,  the  sum of $\bg_{lk}[n]$ expresses the squared modulus of a linear combination of $K$ complex normal random variables, i.e., $g_{lk,n}\sim \Gamma (K,1)$.
\end{proof}
\begin{corollary}
When $M=K$, i.e., in the special case of full  SDMA, the upper bound of the coverage probability with RATHIs and channel aging, described by Theorem~\ref{theoremCoverageProbability}, is given by
\begin{align}
\! \!\!p_{c}\left(  T,\lambda_{B},\al,\delta,q \right)\!\le\!\!\lambda_{B}\!\!\int_{l \in  \mathds{R}^2}\!\!& \mathcal{L}_{E_k}\left( l^{-\al}s  \right)\! \mathcal{L}_{I_{ \etv_{\mathrm{t}}}}\!\!\left(l^{-\al} s \right) \nn\\ 
&\times \mathcal{L}_{I_{ \etv_{\mathrm{r}}}}\!\!\left(l^{-\al} s \right)\!    \mathcal{L}_{I_{l}}\!\!\left(s \right)\mathrm{d}l.\!\!
\end{align}
\end{corollary}

\subsection{Average Achievable Rate}\label{AverageAchievableRate1} 
This subsection presents the mean achievable data rate for a typical user under the  proposed general system model with RATHIs and imperfect CSIT. Given that the analysis and some definitions are similar to Section~\ref{coverage}, the presentation is more concise.  Actually, the following theorem is one of the main results of this paper, being unique in the research area of practical systems with hardware impairments, when the BSs are randomly positioned.

\begin{theorem}\label{AverageAchievableRate}
 The downlink average achievable user  rate of a  HCN in the presence of RATHIs and imperfect CSIT due to pilot contamination and channel aging is
\begin{align}
 &R_{k}\left( T,\lambda_{B},\al,\delta, q \right)=\int_{\!x \in  \mathds{R}^2}\!\! \int_{\!t>0} \!\sum_{i=0}^{\Delta\!-\!1}\!\sum_{k=0}^{i}\!\sum_{u_1+u_2+u_3=i-k}\!\!\!\binom{\!i\!}{\!k\!}\!\nn\\
 &\times\!\binom{\!i\!-\!k\!}{\!u_1\!+\!u_2\!\!+u_3\!}\!\frac{\left(\! -1\! \right)^{i}s^{k}\!}  {i!}\!\frac{\mathrm{d}^{u_1}}{\mathrm{d}s^{u_1}}\mathcal{L}_{E_k}\left( l^{-\al }s \left(e^{t}-1  \right) \right) \nn\\
 &\times\frac{\mathrm{d}^{u_2}}{\mathrm{d}s^{u_2}}\mathcal{L}_{I_{ \etv_{\mathrm{t}}}}\!\!\left(l^{-\al}s\left(e^{t}-1  \right) \right)\! 
                 \frac{\mathrm{d}^{u_3}}{\mathrm{d}s^{u_3}}\mathcal{L}_{I_{ \etv_{\mathrm{r}}}}\!\!\left(l^{-\al} s\left(e^{t}-1  \right) \right)\!\nn\\
 &\times   \frac{\mathrm{d}^{k}}{\mathrm{d}s^{k}} \mathcal{L}_{I_{l}}\!\!\left(s\left(e^{t}-1  \right) \right)\!
\mathrm{d}t\mathrm{d}l,
 \end{align}
 where the various parameters are also defined in Theorem~\ref{theoremCoverageProbability}.
\end{theorem}
\begin{proof}
See Appendix~\ref{AverageAchievableRateproof}.
\end{proof}
 \section{Numerical Results}  \label{Numerical} 
This section aims at investigating the impact of the various parameters such as number of users and antennas on the general expressions corresponding to the coverage probability and the user rates provided by Theorems~\ref{theoremCoverageProbability} and~\ref{AverageAchievableRate}, respectively.
Given that the typical user lies at the origin, we choose a sufficiently large area of $50~\mathrm{km}\times 50$ $\mathrm{km}$, where  the locations of the BSs are simulated as realisations of a  PPP with given density $\lambda_{B}=0.01$. The users' PPP density is considered to be $\lambda_{k}=60 \lambda_B$ as in~\cite{Bai2016}.

Coverage of the user means that the received SIR from at least one of the BSs  exceeds a certain target. We are interested  in the calculation of the desired signal strength and interference power. Note that the received SIR is obtained from each BS. The desired estimate of the coverage probability comes by repeating this
scheme an adequate number of times. Thus, this procedure allows us not only to validate our model, but also to demonstrate the effect of the various parameters on the coverage probability. 

Herein, we consider a setup, where the number of antennas per BS and and the number of users are $M=5$ and $K=3$, respectively. The path-loss is set to $\al=3$, while the uplink and downlink transmit powers are $\rho^{\mathrm{UE}}_{\mathrm{up}}=5~\mathrm{dB}$ and $p=15~\mathrm{dB}$. Note that  $\tau=K$, hence $p_{\mathrm{tr}}=15~\mathrm{dB}$. Moreover, we assume that the distance between the user and the BS during the uplink phase is $\| x_{k}\|$=15~$\mathrm{m}$, and it is known by the BS.   In  the figures,  the proposed analytical expressions of the coverage probability $p_{c}\left( T,\lambda_{B},\al,\delta, q \right)$ and the achievable user rate $R_{k}\left( T,\lambda_{B},\al,\delta, q \right)$ are plotted  along with the corresponding simulated results. The ``dot'' and ``solid'' lines illustrate the analytical results with specific RATHIs/user mobility and no transceiver impairments or no relative user movement. In a similar way, the bullets designate the simulation results. Obviously, the inevitable presence of RATHIs or the time variation of the channel result in the worsening of the system performance. Actually, the more severe these effects are, the higher the degradation of the system performance becomes.
\subsection{Impact of RATHIs} 
In order to investigate how the RATHIs affect the coverage probability, we assume no user mobility, i.e., $\delta=0$. First, we focus on the uplink RATHIs $\kappa_{\mathrm{r}_{\mathrm{BS}}}^{2}$ and $\kappa_{\mathrm{t}_{\mathrm{UE}}}^{2}$, while we have assumed that the downlink RATHIs are zero. Specifically, in Fig~\ref{Fig5}, we show the variation of the coverage probability for different values of uplink hardware impairments. Based on this model, the transmit impairment of the user and the receive impairment at the BS contribute the same at the coverage probability, which is in agreement with theory. For example, we set $\kappa_{\mathrm{t}_{\mathrm{UE}}}^{2}=0.08$ and $\kappa_{\mathrm{r}_{\mathrm{BS}}}^{2}=0$  as well as $\kappa_{\mathrm{t}_{\mathrm{UE}}}^{2}=0$ and $\kappa_{\mathrm{r}_{\mathrm{BS}}}^{2}=0.08$ and the coverage probability does not change. Thus, as a design plan, we could keep $\kappa_{\mathrm{t}_{\mathrm{UE}}}^{2}$ constant and play with the quality of the BS hardware for obtaining a certain coverage probability.
 \begin{figure}[!h]
 \begin{center}
 \includegraphics[width=\linewidth]{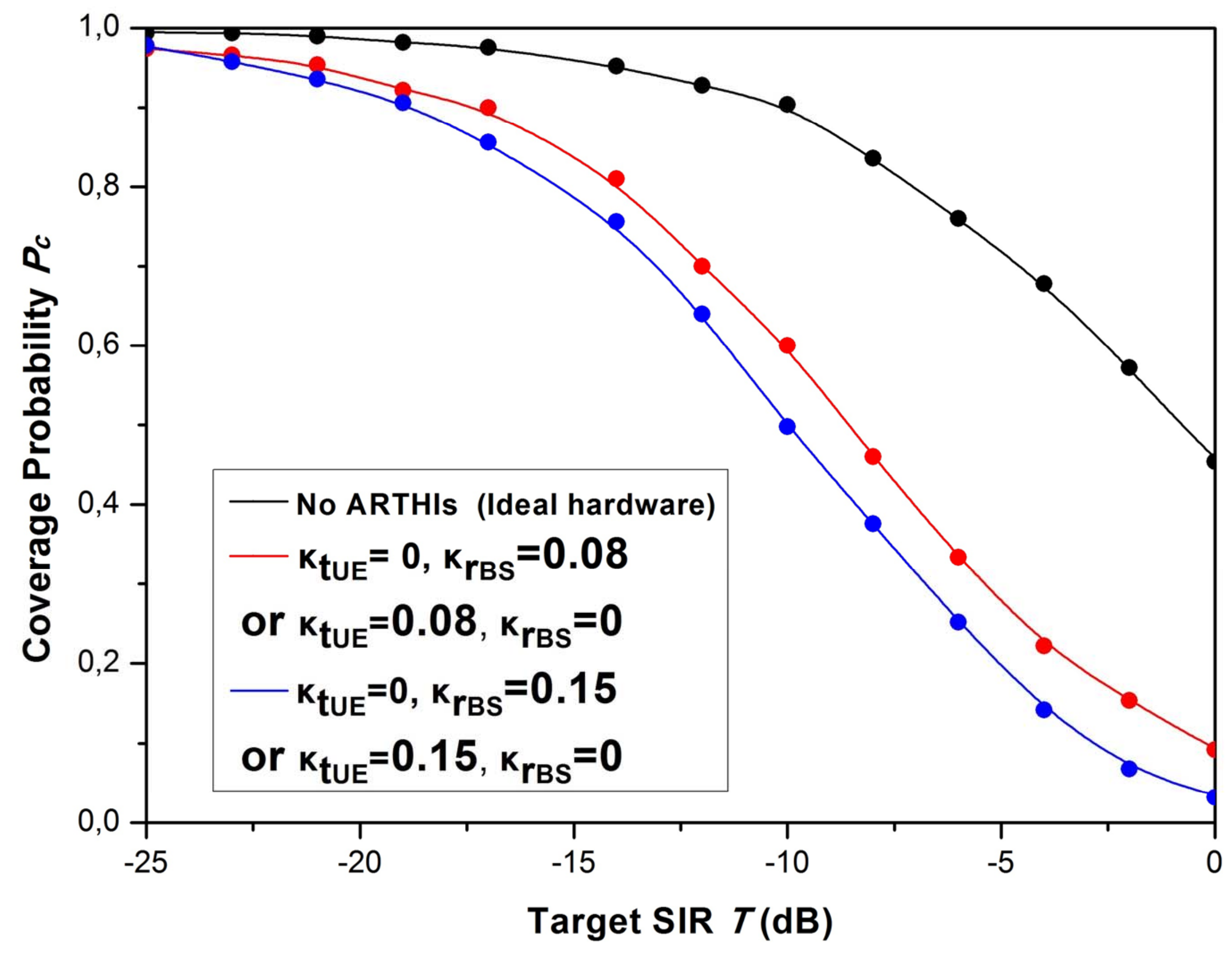}
 \caption{\footnotesize{Coverage probability of a multiple-antenna HCN for varying uplink RATHIs versus $\rho$  ($M=5$, $K=3$, $\al=3$, $\kappa_{\mathrm{t}_{\mathrm{BS}}}=\kappa_{\mathrm{r}_{\mathrm{UE}}}=0$,  $\delta=0$).}}
 \label{Fig5}
 \end{center}
 \end{figure}
  \begin{figure}[!h]
 \begin{center}
 \includegraphics[width=\linewidth]{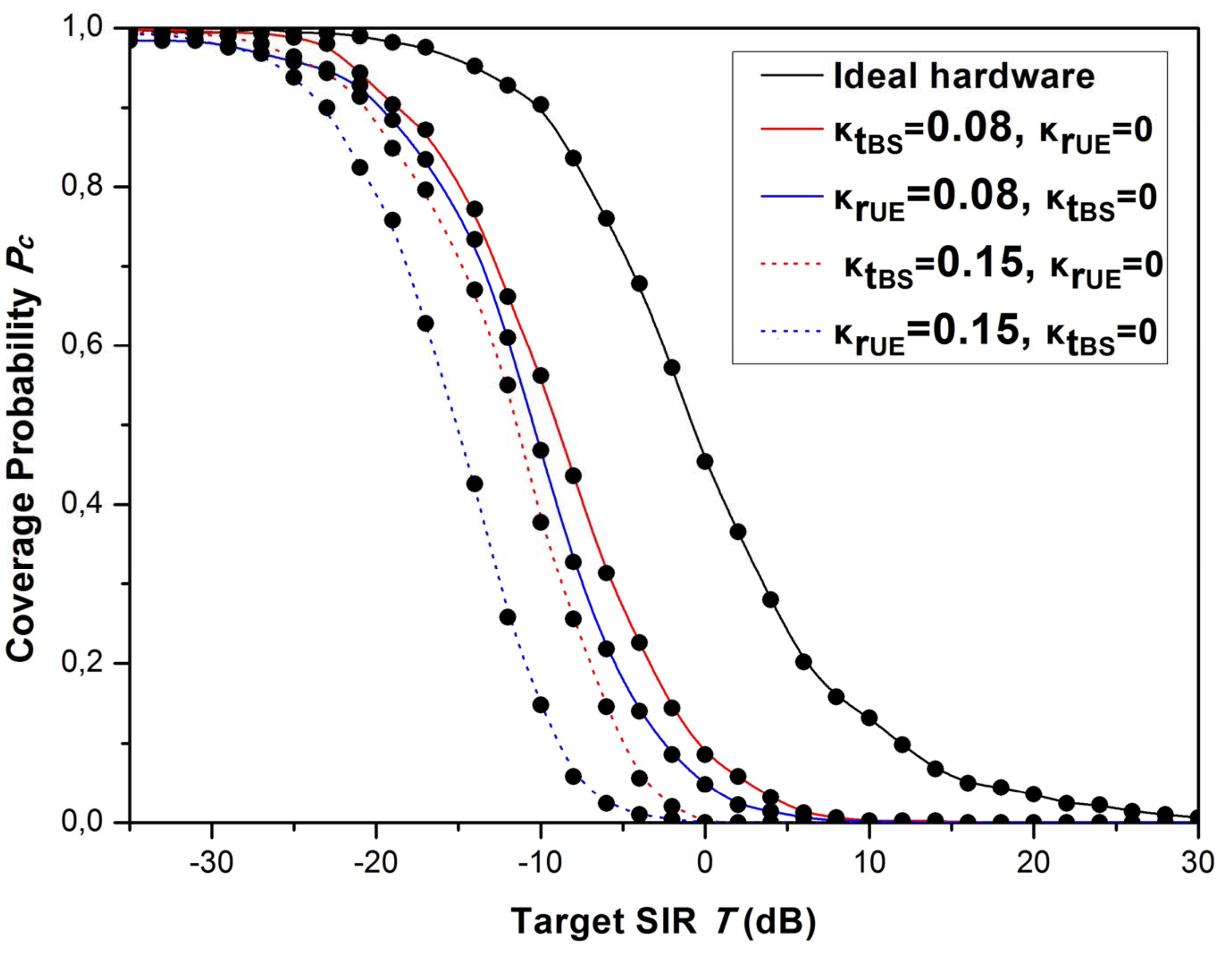}
 \caption{\footnotesize{Coverage probability of a multiple-antenna HCN for varying downlink RATHIs versus $\rho$  ($M=5$, $K=3$, $\al=3$, $\kappa_{\mathrm{t}_{\mathrm{UE}}}=\kappa_{\mathrm{r}_{\mathrm{BS}}}=0$,  $\delta=0$).}}
 \label{Fig2}
 \end{center}
 \end{figure} 
\begin{figure}[!h]
 \begin{center}
 \includegraphics[width=\linewidth]{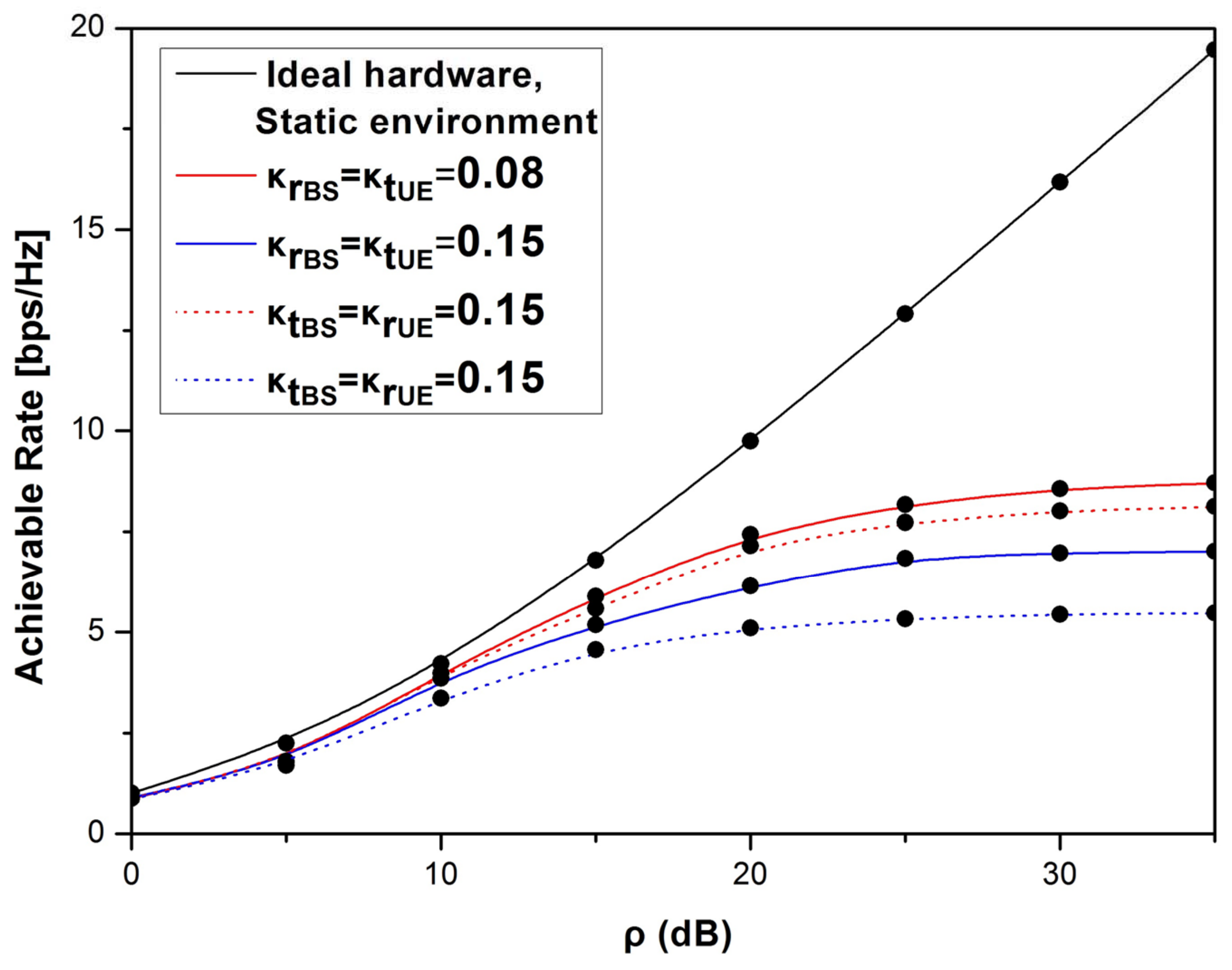}
 \caption{\footnotesize{Achievable user rate of a multiple-antenna HCN for varying uplink and downlink RATHIs versus $\rho$  ($M=5$, $K=3$, $\al=3$,   $\delta=0$).)}}
 \label{Fig1}
 \end{center}
 \end{figure} \begin{figure}[!h]
 \begin{center}
 \includegraphics[width=\linewidth]{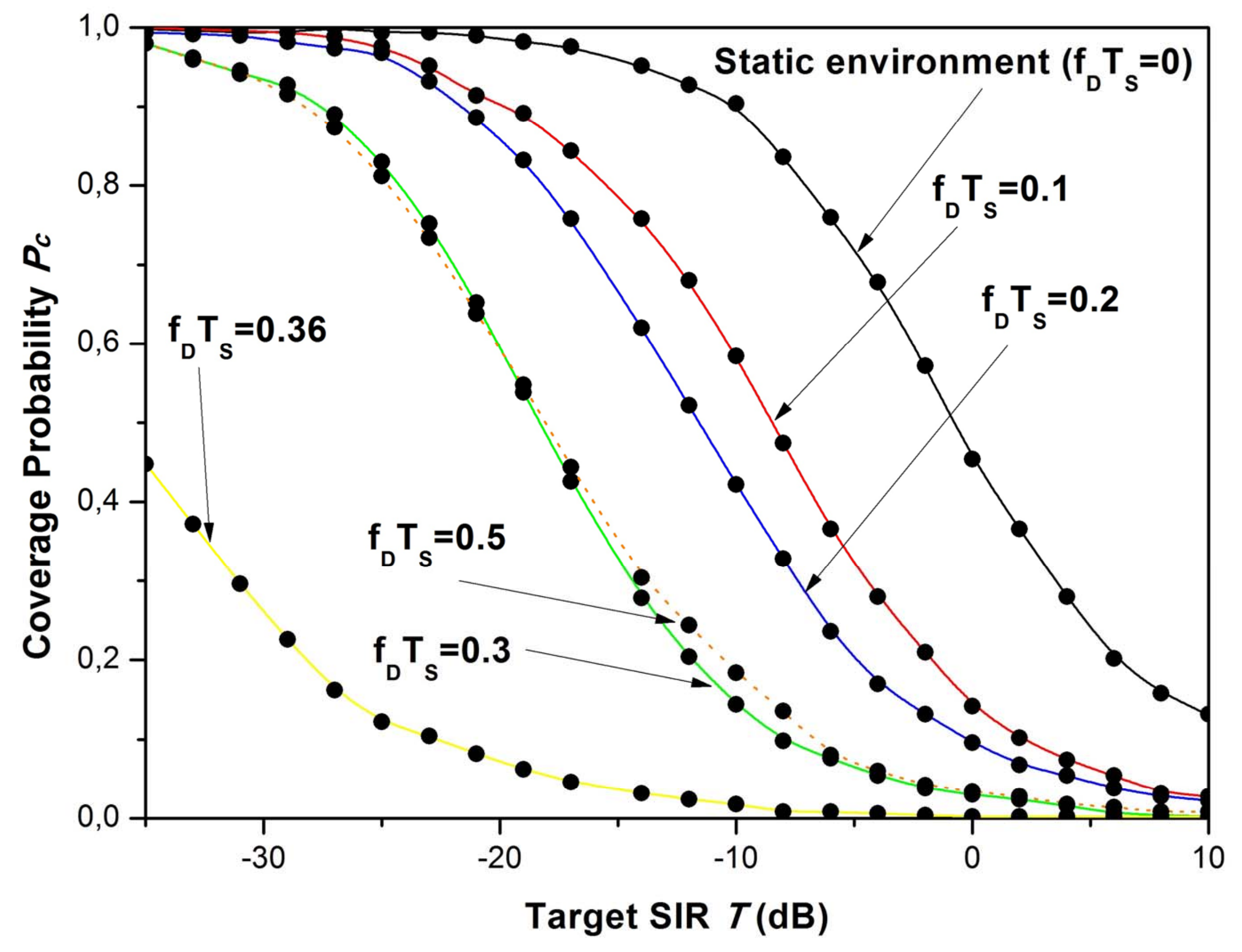}
 \caption{\footnotesize{Achievable user rate of a HCN versus $\rho$   for varying user mobility ($M=5$, $K=3$, $\al=3$, $\kappa_{\mathrm{t}_{\mathrm{UE}}}=\kappa_{\mathrm{r}_{\mathrm{BS}}}=\kappa_{\mathrm{r}_{\mathrm{UE}}}=\kappa_{\mathrm{r}_{\mathrm{UE}}}=0$,  $\delta=0$).}}
 \label{Fig3}
 \end{center}
 \end{figure} \begin{figure}[!h]
 \begin{center}
 \includegraphics[width=\linewidth]{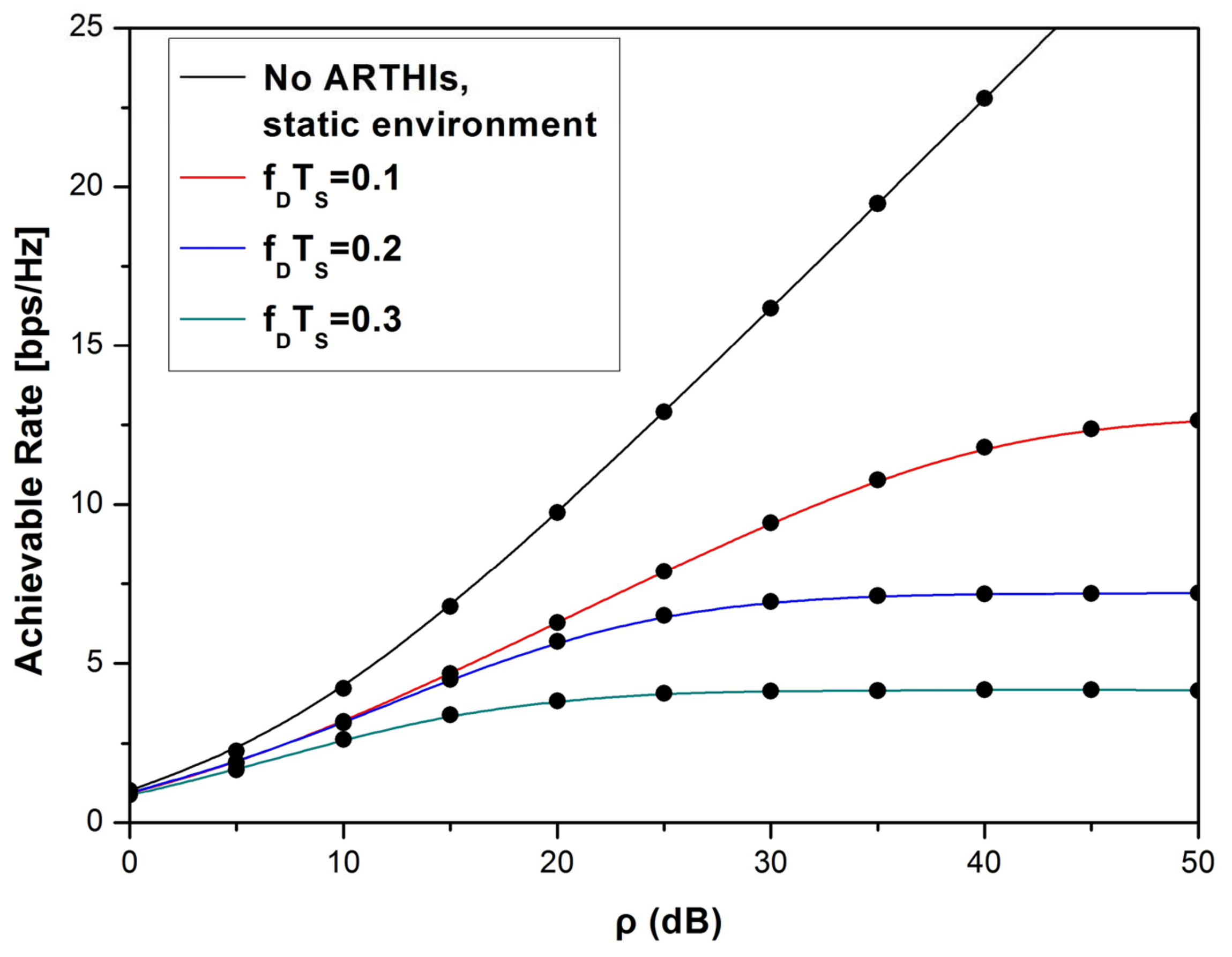}
 \caption{\footnotesize{Coverage probability of a multiple-antenna HCN versus $\rho$   for varying Doppler shift ($M=5$, $K=3$, $\al=3$, $\kappa_{\mathrm{t}_{\mathrm{UE}}}=\kappa_{\mathrm{r}_{\mathrm{BS}}}=\kappa_{\mathrm{r}_{\mathrm{UE}}}=\kappa_{\mathrm{r}_{\mathrm{UE}}}=0$,  $\delta=0$). }}
 \label{Fig4}
 \end{center}
 \end{figure} 
  \begin{figure}[!h]
 \begin{center}
 \includegraphics[width=\linewidth]{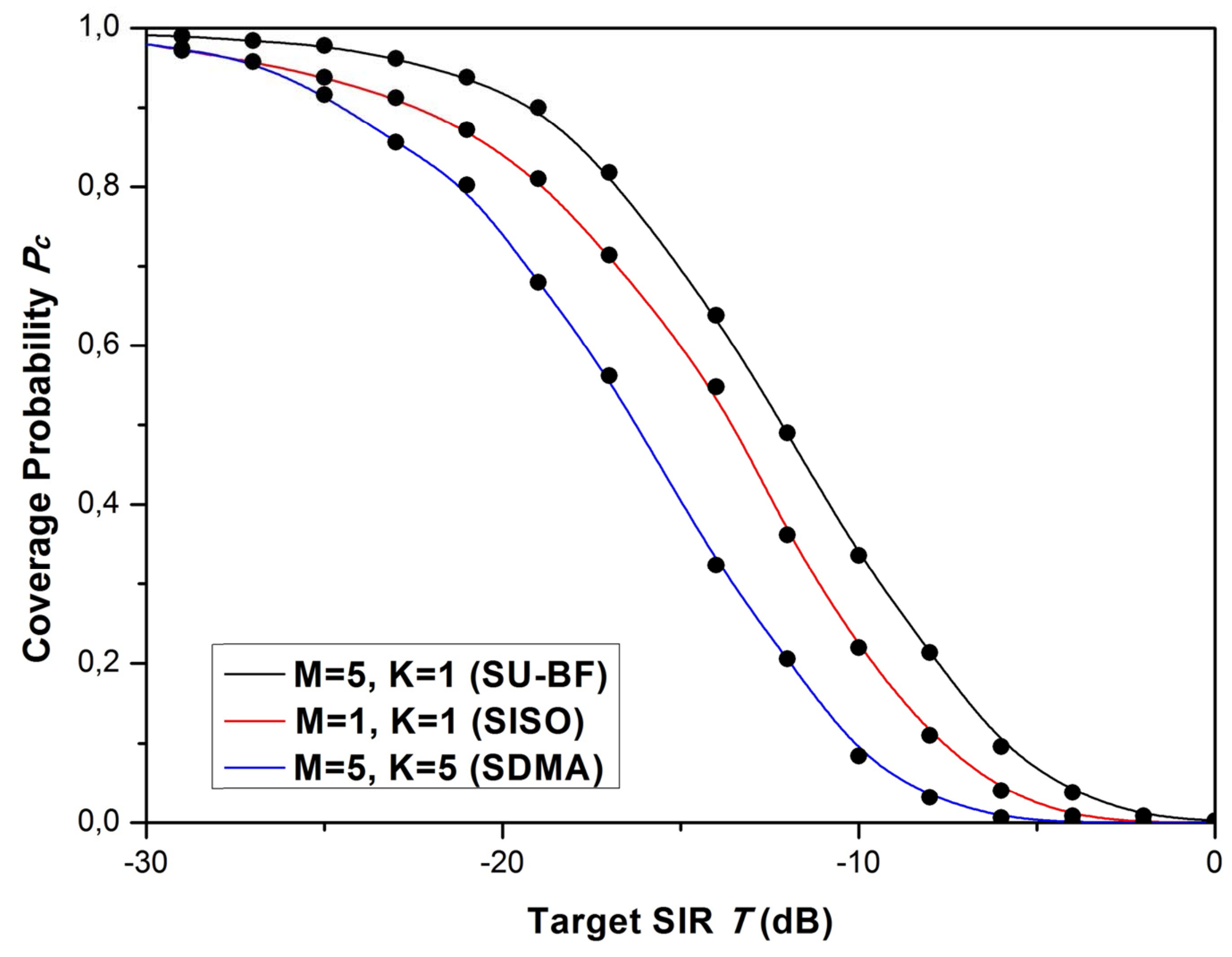}
 \caption{\footnotesize{Coverage probability of a multiple-antenna HCN versus $\rho$ for various
combinations of multi-antenna techniques   ($\al=3$, $\kappa_{\mathrm{t}_{\mathrm{UE}}}=\kappa_{\mathrm{r}_{\mathrm{BS}}}=\kappa_{\mathrm{r}_{\mathrm{UE}}}=\kappa_{\mathrm{t}_{\mathrm{BS}}}=0.08$,  $\delta=0.1$). }}
 \label{Fig6}
 \end{center}
 \end{figure} 
 
In Fig.~\ref{Fig2}, we plot the coverage probability as a function of the target SIR for different values of the downlink RATHIs. Clearly, here the transmit hardware impairment  $\kappa_{\mathrm{t}_{\mathrm{BS}}}$ exhibits higher impact on $p_{c}\left( q \right)$ than $\kappa_{\mathrm{r}_{\mathrm{UE}}}$, which coincides with the theoretical observations.  For this reason, we should be more careful with the quality of the transmit impairments at the BS.  These nominal values
of RATHIs are quite reasonable according to~\cite{Bjornson2014}. 
Moreover, in the same figure, we have depicted the simulated and
theoretical results corresponding to ideal hardware for the sake of comparison. Similar conclusions can be extracted by observing the rate versus $\rho$ for varying downlink RATHIs in Fig~\ref{Fig1}.
\subsection{Impact of User Mobility}
We study the effect of time-variation of the channel due to user mobility on the coverage
probability of a cellular network in Fig.~\ref{Fig3}.
These observations are consistent with the coverage probability and user rate results derived in Sections~\ref{coverage} and~\ref{AverageAchievableRate1}, respectively. Note that in order to focus only on the effect of time-variation of the channel, we assume that the hardware is ideal, i.e.,  both in the downlink and the uplink the RATHIs are set to zero. Clearly, it is shown that an increase of the time variation of the channel results in a decrease of $p_{c}\left( T,\lambda_{B},\al,\delta, q \right)$. This behaviour continues till the normalised Doppler shift $f_{D}T_{s}$ becomes $f_{D}T_{s}\approx 0.36$. Then, increasing $f_{D}T_{s}$, the coverage probability is improved. Actually, behind this behaviour is hidden the variation of the Jakes autocorrelation function with $f_{D}T_{s}$. In Fig.~\ref{Fig4}, we illustrate the achievable rate versus $\rho$. It is shown that the rate decreases with increasing time-variation, as expected.
\subsection{Comparison between Transmission Techniques}
We elaborate on the following three transmission setups: i) SISO with $M=K=1$, i.e., we have a single transmit antenna per BS and one user ii) full SDMA with $M=K=5$, i.e., $5$ transmit antennas per BS serving $5$ users iii) SU-BF transmission, in order to shed light on the realistic performance of HCNs, as far as the system parameters $M$ and $K$ change. Furthermore, we choose a set of parameters of RATHIs and time variation, representing a realistic scenario,  to be $q=\{0.03,0.03,0.08,0.08\}$ and $\delta=0.1$. 

In Fig.~\ref{Fig6}, we observe that the SU-BF transmission achieves higher coverage probability with comparison to SISO. Note that the latter is better than SDMA. In other words, we verify that it is better to serve
a single user in each resource block, either by SISO or SU-BF, instead of serving multiple users. However, herein and with comparison to~\cite{Dhillon2013}, we illustrate how, given this property, the coverage probability varies with SIR in the presence of the RATHIs. Similar observations have been mentioned in~\cite{Papazafeiropoulos2016b} and in~\cite{Dhillon2013} for the cases of  RATHIs and perfect CSIT with ideal hardware, respectively.
Especially, we illustrate that SU-BF transmission
proves to be better with comparison to SISO in some regimes/situations because, in  addition to the proximity gain enjoyed by the SISO due to extreme densification, the SU-BF transmission presents an additional beamforming
gain. Moreover, SISO outperforms SDMA. More concretely, instead of serving multiple users, it is preferable to serve
a single user in each resource block by means of  SISO or SU-BF. Herein, it is worthwhile to mention that in interference-limited networks, the use of more antennas for multi-stream transmission (SDMA) is not always beneficial, but it depends on which transmission/reception scheme is employed. Also, this claim depends on the
performance metric we study.    Interestingly, having this property in mind, we depict how the  presence of the
RATHIs affect the performance of the system in terms of the degradation of both coverage probability and user rate.

\section{Conclusion} \label{Conclusion} 
In this paper, we presented a new framework for the downlink
cellular network analysis with imperfect CSIT. Specifically, we proposed a multiple antenna HCN that generalises the state of the art by accounting for imperfect CSIT due to pilot contamination, channel aging, and RATHIs. Furthermore, we quantified the performance loss due to imperfect CSIT. We showed that the uplink RATHIs have equivalent contribution to both metrics under investigation, while the downlink RATHIs exhibit different impact. Specifically, the quality of the BS transmitter has greater impact than the receive hardware of the user. Moreover, we demonstrated the loss due to time variation in both coverage probability and user rate. In addition, we demonstrated the outperformance of SU-BF in terms of coverage and user-rate against both SISO and SDMA in the case of a more realistic analysis than previous works did. In other words, we confirmed that, given a total number of transmit antennas, it is better to share them across many BSs than collecting them in a few BSs. Actually, the proposed framework is of high interest because it allows us to explore how various multiple antenna techniques affect the coverage and the rate in realistic conditions. \begin{appendices}
\section{Proof of Proposition~\ref{SINR}}\label{SINRproof}
Initially, we assume  that the columns of the precoding matrix $\hat{\bW}[n]$ with unit norm equal the normalised columns of $\hat{\bH}^{\H}[n]\left( \hat{\bH}[n] \hat{\bH}^{\H}[n] \right)^{-1}$. In other words, $\hat{\bW}[n]=\bar{\bH}^{\H}[n]\left( \bar{\bH}[n] \bar{\bH}^{\H}[n] \right)^{-1}$, where $\bar{\bH}[n]=\left[ \bar{\bh}_{1}[n],\ldots, \bar{\bh}_{K}[n]\right]\in \mathbb{C}^{\left( M \times K \right)}$ with columns $\bar{\bh}[n]=\frac{\hat{\bh}[n]}{\|\hat{\bh}[n]\|}$. As a result, the numerator of the SIR in~\eqref{eq: general sum_rate}, describing the desired signal power, is $\Gamma\left( \Delta,\sigma_{\hat{\bh}_{k}}^{2} \right)$ distributed with $\Delta=M-K+1$ (or else it follows the Erlang distribution with  shape and  scale parameters $\Delta$ 
and $\sigma_{\hat{\bh}_{k}}^{2}$, respectively), since  $Z_{k}[n]=|\bar{\bh}_{k}^{\H}[n]\bw_{k}[n]|^{2}\cdot \|\bh_{k}[n]\|^{2}$. Actually,  $Z_{k}[n]$  can be written as the product of two independent random variables distributed as $B\left( M-K+1,K-1 \right)$ and $\Gamma\left( M,\sigma_{\hat{\bh}_{k}}^{2} \right)$, respectively~\cite{Jindal2006}\footnote{The random variable $\|\hat{\bh}_{k}^{\H}[n]\|^{2}$ is the linear combination of $M$ i.i.d. exponential random variables each with variance $ \sigma_{\hat{\bh}_{k}}^{2}$, i.e., $\|\hat{\bh}_{k}^{\H}[n]\|^{2}\mathop \sim \limits^{\tt d}\Gamma[M,\sigma_{\hat{\bh}_{k}}^{2}]$.}.  Note that ZF beamforming has been applied. Therefore, the resultant PDF of the product follows the $\Gamma\left( \Delta,\sigma_{\hat{\bh}_{k}}^{2} \right)$ distribution. 
The term in the denominator, concerning the error, is expressed in terms of a sum of $K$ independent random variables   as
 \begin{align}
 \!\!\!\!\! \!E_k[n]&=\delta^{-2} \left(1+\kappa_{\mathrm{t}_{\mathrm{BS}}} ^{2}\right)\big\|
                    \tilde{\bee}_{k}^{\H}[n]\hat{\bW}[n\!-\!1]
              \big\|^2\!\!\nn\\
              &=\!\! =\delta^{-2} \left(1+\kappa_{\mathrm{t}_{\mathrm{BS}}} ^{2}\right)\sum_{i=1}^{K} \!\big|
                    \tilde{\bee}_{k}^{\H}[n]\hat{\bw}_{i}[n\!-\!1]
              \big|^2\nn.
 \end{align}
Since $\big|  \tilde{\bee}_{k}^{\H}[n]\hat{\bw}_{i}[n\!-\!1]
              \big|^2$ is the squared modulus of a linear combination of $M$ complex random variables with $\hat{\bw}_{k}[n\!-\!1]$ being independent and having unit norm, it is  distributed as $\Gamma\left( 1,\sigma_{\hat{\bee}_{k}}^{2} \right)$. Thus, $\big\|
                    \tilde{\bee}_{k}^{\H}[n]\hat{\bW}[n\!-\!1]
              \big\|^2$ is  $\Gamma\left( K,\sigma_{\hat{\bee}_{k}}^{2} \right)$ distributed.
Taking the expectation over the transmit distortion noise of the tagged BS, then   $I_{ \etv_{\mathrm{t}}^{\mathrm{BS}}}[n]={p\kappa_{\mathrm{t}_{\mathrm{BS}}}^{2}}\|{\bh}_{k}^{\H}[n]\|^2$, which follows a scaled $\Gamma (M,\sigma_{\hat{\bh}_{k}}^{2})$ distribution. Following a similar procedure, we take the expectation over the receive distortion noise $\etv_{\mathrm{r},n}^{\mathrm{BS}}$. We obtain $I_{ \etv_{\mathrm{r}}^{\mathrm{UE}}}[n]=p{\kappa_{\mathrm{r}_{\mathrm{UE}}}^{2}}\|\bh_{k}[n]\|^2$ following a scaled $\Gamma (M,1)$. Regarding the other term in the denominator, it represents  the interference from other BSs, $I_{l}[n]$. Especially,  $g_{lk,n}=|\bg_{lk}^{\H}[n]\bW_{l}[n]|^2\sim \Gamma (K,1)$ because $\bW_{l}[n-1]$ being the precoding matrices from other BSs have unit-rorm and are independent from the normalised $\bar{\bg}_{lk}[n]$. On this  ground, $g_{lk,n} $  is a linear combination of $K$ independent complex normal random variables with unit variance, i.e., $g_{lk,n}\sim \Gamma (K,1)$.
\section{Proof of Theorem~\ref{theoremCoverageProbability}}\label{CoverageProbabilityproof}
Starting from the definition of $p_{c}\!\left(  T,\lambda_{B},\al,\delta,q\right)$ and by means of appropriate substitution of the SIR $\gamma_{k}$, we have
\begin{align}
 \!\!&p_{c}\!\left(  T,\lambda_{B},\al,\delta,q\right)=\EE\!\left[\mathds{1}\!\left( 	\underset{x_{k,n} \in \Phi_{B}}{\cup} \mathrm{SIR}\left( x_{k,n} \right)>T \right)  \right]\\
 &\le\EE\!\left[ 	\underset{x_{k,n} \in \Phi_{B}}{\cup} \mathds{1}\left(  \mathrm{SIR} \right)>T   \right]\label{coverage_definition0}\\
\! \!&=\!\!\EE \!\left[\sum_{x_{k,n} \in \Phi_{B}}\!\!\mathbb{P}\left[ \mathrm{SIR}>T|l \right]\right]\label{coverage_definition}\\
                                   \! \!&=\!\!  \lambda_{B}\!\!\!\int_{x \in  \mathds{R}^2}\!\!\!\!\!\!\!\EE\left[ \mathbb{P}\!\left[   Z_{k}\!>\!
                                  {T} \! \left(E_k[n]\!+\! I_{ \etv_{\mathrm{t}}}[n]\!+\!I_{ \etv_{\mathrm{r}}}[n] \right) +  {T} l^{\al} I_{l}|l \right]\right]\!\mathrm{d}x,\label{coverage_definition1} 
\end{align}
where in~\eqref{coverage_definition0}, we have employed the union bound or  else the Boole's inequality, and in~\eqref{coverage_definition1}, we have used the Campbell-Mecke Theorem~\cite{Chiu2013a}.  Given that $Z_{k}[n]$ is Gamma distributed, its PDF is provided by~\eqref{eq PDF1 1}. Below, we have removed the index $n$ for the ease of exposition.
We turn our attention to the integrable part of~\eqref{coverage_definition1}. This can be expressed   as
\begin{align}
 &\mathbb{P}\!\left[\!Z_{k}\!>\!{T} \! \left(E_k\!+\! I_{ \etv_{\mathrm{t}}}\!+\!I_{ \etv_{\mathrm{r}}} \right) \!+\!  {T} l^{\al} I_{l}|l \right]\!=\!e^{- \frac{T \left( E_k+I_{ \etv_{\mathrm{t}}}+I_{ \etv_{\mathrm{r}}} \right)}{\sigma_{\hat{\bg}_{k}}^{2}} }
\nn\\
&\!\times\! e^{-\frac{{T} l^{\al} I_{l}}{\sigma_{\hat{\bg}_{k}}^{2}}}\!\sum_{i=0}^{\Delta-1}\sum_{k=0}^{i}\!\binom{i}{k}\frac{\left(  {T}  \left( E_k\!+\!I_{ \etv_{\mathrm{t}}}\!+\!I_{ \etv_{\mathrm{r}}} \right)  \right)^{i-k}\!\left(     {T} l^{\al} I_{l}  \right)^{k}}{i!\left( \sigma_{\hat{\bg}_{k}}^{2} \right)^{i}},\label{coverage1}
                                \end{align}
                                where in~\eqref{coverage1}, we have applied the Binomial theorem. In addition, if we  apply the expectation operator, we obtain 
\begin{align}
 &\!\EE\!\left[\!\mathbb{P}\!\left[Z_{k}\!>\!{T}\! \left( I_{ \etv_{\mathrm{t}}}\!+\!I_{ \etv_{\mathrm{r}}} \right)\! + \! {T} l^{\al} I_{l}|l \right]\right]\!=\!\!
                    \sum_{i=0}^{\Delta-1}\!\sum_{k=0}^{i}\!\sum_{u_1+u_2+u_3=i-k}\!\!\binom{i}{k}\!\nn\\
                  &\!\times\binom{i\!-\!k}{u_1+u_2+u_3}\!\frac{\left( -1 \right)^{i}\!s^{k}  }{i!}\frac{\mathrm{d}^{u_1}}{\mathrm{d}s^{u_1}}\mathcal{L}_{E_k}\left( s \right)\frac{\mathrm{d}^{u_2}}{\mathrm{d}s^{u_2}}\mathcal{L}_{I_{ \etv_{\mathrm{t}}}}\!\!\left(s \right)\!  \nn\\
                  &\times  \frac{\mathrm{d}^{u_3}}{\mathrm{d}s^{u_3}}\mathcal{L}_{I_{ \etv_{\mathrm{t}}}}\!\!\left(s \right)\!   \frac{\mathrm{d}^{k}}{\mathrm{d}s^{k}} \mathcal{L}_{I_{l}}\!\!\left(s \right)\!,\!\label{coverage3}
\end{align}
where we have set  $ s=\frac{{T}}{\sigma_{\hat{\bg}_{k}}^{2}} l^{\al}$. We result in~\eqref{coverage3}, after using the Multinomial theorem. Note that the inner sum is taken over all combinations of nonnegative integer indices $u_1$ through $u_3$ such that the sum  $u_1+u_2+u_3$ is $i-k$. Moreover, we have used the definition of the Laplace Transform $\mathbb{E}_{I }\left[  e^{-s I }\left( s I  \right)^{i}\right]=s^{i}\mathcal{L}\{t^{i}g_{I }\left( t \right)\}\left( s \right)$ and the Laplace identity $t^{i}g_{I }\left( t \right)\longleftrightarrow \left( -1 \right)^{i}\frac{\mathrm{d}^{i}}{\mathrm{d}^{i}s}\mathcal{L}_{I }\{g_{I }\left( t \right)\}\left( s \right)$. 
Regarding the Laplace transform  $\mathcal{L}_{I_{l}}\left( s \right)$, it is obtained by means of Proposition~\ref{LaplaceTransform}, while the transforms $\mathcal{L}_{I_{ \etv_{\mathrm{t}}}}$ and    $\mathcal{L}_{I_{ \etv_{\mathrm{r}}}}\!\left(s \right)$ are provided by Lemma~\ref{LaplaceTransformGamma}. Finally, the proof is concluded after substituting~\eqref{coverage3} into~\eqref{coverage_definition1}.              
\section{Proof of Proposition~\ref{LaplaceTransform}}\label{LaplaceTransformproof}
The Laplace transform of the interference power relies on its PDF. Specifically,  setting $g_{l}\triangleq g_{lk,n}$, accounting for  the power of the interference channel, has  identical distribution for all $l$,  $\mathcal{L}_{I_{l}}\left( s \right)$ can be obtained as
\begin{align}
 &\!\!\mathcal{L}_{I_{l}}\!\left(  s \right)=\mathbb{E}_{I_{l}}\!\left[  e^{-s {I_{l}}}\right]=\mathbb{E}_{I_{l}}\!\left[  e^{-s \sum_{l\in \Phi_{B}\backslash x} g_{l} y^{-\al}}\right]\nn\\
  &\!\mathop = \mathbb{E}_{\Phi_{B},g_{l}}\!\left[\prod_{l \in \Phi_{B}\backslash x} e^{-s  g_{l} y^{-a}} \right] \label{laplace 2}\\
&\!\mathop = \mathbb{E}_{\Phi_{B}}\!\left[\prod_{l \in \Phi_{B}\backslash x}  \mathcal{L}_{g_{l}}\left( s  y^{-a} \right)  \right]\label{laplace 3}\\
  &\!\mathop =   \mathrm{exp}\!\left( -{\lambda_{B}}\int_{\mathbb{R}^{2}} \left( 1-\mathcal{L}_{g_{l}}\left( s  y^{-a} \right)   \right)\mathrm{d}y\right)\label{laplace 4}\\
&\! \!\mathop =\mathrm{exp}\!\Bigg(\Bigg.\!\!\!-\!2 \pi \lambda_{B}\!\!\!\int_{0}^{\infty}  \!\!\bigg(\bigg. \frac{ \sum_{m=1}^{K}\binom{K}{m} \left( {s}  r^{-a} \right)^{m}}{\left( 1\!+\!{s}  r^{-a} \right)^{K}}r\mathrm{d}r\Bigg)\!\!\label{laplace5} \\
&\!\!\! \!\mathop =\mathrm{exp}\!\Bigg(\Bigg.\!\!-\frac{2 \pi \lambda_{B} {s}^{\frac{2}{a}}}{\al}\!\!\sum_{m=1}^{K}\!\!\binom{\!K\!}{\!m\!}\mathrm{B}\!\left(\! K\!-\!m\!+\!\frac{2}{a},m\!-\!\frac{2}{a} \right)\!\!\!\!\Bigg).\!\!\!\nn
\end{align}
In particular,~\eqref{laplace 2} results from the independence among the  locations of the BSs. Moreover~\eqref{laplace 2} holds due to the independence between the spatial and the fading distributions. The property of the probability generating functional (PGFL) regarding the PPP~\cite{Chiu2013} is used to obtain~\eqref{laplace 4}. Furthermore, we substitute the Laplace transform of $g_{l}$ following a Gamma distribution. Note that in the next step, application of the Binomial theorem takes place, while we continue with conversion of the Cartesian coordinates to polar coordinates in~\eqref{laplace5}. Finally, we conclude the proof by  calculating the integral in the following way. We make the substitution  $\left( 1+r^{-\al} \right)^{-1}\rightarrow t$, and after many algebraic manipulations, and the use of~\cite[Eq.~(8.380.1)]{Gradshteyn2007} we lead to the desired result.
\section{Proof of Theorem~\ref{AverageAchievableRate}}\label{AverageAchievableRateproof} 
The achievable rate of the typical user at the origin is derived as
\begin{align}
 &R_{k}\left(  T,\lambda_{B},\al,\delta,q \right)=\EE\left[\ln \left( 1+\mathrm{SIR} \right)\right]	\label{rate1} \\
 &=\int_{x \in  \mathds{R}^2} \EE \left[\ln \left( 1+\mathrm{SIR} \right)\right]\mathrm{d}l\label{rate2}\\
  &=\int_{x \in  \mathds{R}^2} \int_{t>0} \mathbb{P} \left[\ln \left( 1+\mathrm{SIR} \right)\right]\mathrm{d}t\mathrm{d}l\label{rate3}\\
  &=\int_{x \in  \mathds{R}^2} \int_{t>0}   \mathbb{P}\!\left[h_{k}\!>\!\left({T} \left( I_{ \etv_{\mathrm{t}}}+I_{ \etv_{\mathrm{r}}} \right) +  {T} l^{\al} I_{l}\right)\left(e^{t}-1  \right) \right]
\mathrm{d}t\mathrm{d}l\label{rate4}\\
&=\int_{x \in  \mathds{R}^2} \int_{t>0}
\EE\big[ \!e^{- \frac{{T}}{\sigma_{\hat{\bg}_{k}}^{2}} \left( E_k+I_{ \etv_{\mathrm{t}}}+I_{ \etv_{\mathrm{r}}} \right)\left(e^{t}-1  \right) }e^{-\frac{{T}}{\sigma_{\hat{\bg}_{k}}^{2}} l^{\al} I_{l}\left(e^{t}-1  \right)}
\nn\\
&\!\times\! \!\sum_{i=0}^{\Delta-1}\!\sum_{k=0}^{i}\!\!\binom{\!i}{\!k}\!\frac{\left(  {T} \! \left( E_k\!+\!I_{ \etv_{\mathrm{t}}}\!+\!I_{ \etv_{\mathrm{r}}} \right)\!\left(e^{t}\!-\!\!1  \right)  \right)^{i-k}\!\left(     {T} l^{\al} I_{l} \!\left(e^{t}\!\!-1  \right) \!\right)^{k}}{i! \left( \sigma_{\hat{\bg}_{k}}^{2} \right)^{i}}\big]\!
\mathrm{d}t\mathrm{d}l\label{rate5}\\
&\!=\!\!\int_{\!x \in  \mathds{R}^2}\!\! \int_{\!t>0} \!\sum_{i=0}^{\Delta\!-\!1}\!\sum_{k=0}^{i}\!\sum_{u_1+u_2+u_3=i-k}\!\!\!\binom{\!i\!}{\!k\!}\!\!\binom{\!i\!-\!k\!}{\!u_1\!+\!u_2\!\!+u_3\!}\!\frac{\left(\! -1\! \right)^{i}\!s^{k}  }{i!}\nn\\
                  &\!\times\frac{\mathrm{d}^{u_1}}{\mathrm{d}s^{u_1}}\mathcal{L}_{E_k}\left( l^{-\al }s \left(e^{t}-1  \right) \right)\frac{\mathrm{d}^{u_2}}{\mathrm{d}s^{u_2}}\mathcal{L}_{I_{ \etv_{\mathrm{t}}}}\!\!\left(l^{-\al}s\left(e^{t}-1  \right) \right)\!  \nn\\
                  &\times  \frac{\mathrm{d}^{u_3}}{\mathrm{d}s^{u_3}}\mathcal{L}_{I_{ \etv_{\mathrm{r}}}}\!\!\left(l^{-\al} s\left(e^{t}-1  \right) \right)\!   \frac{\mathrm{d}^{k}}{\mathrm{d}s^{k}} \mathcal{L}_{I_{l}}\!\!\left(s\left(e^{t}-1  \right) \right)\!
\mathrm{d}t\mathrm{d}l,\label{rate6}
 \end{align}
where the expectation in~\eqref{rate1} is taken over the fading distribution and the spatial randomness described by a PPP. In~\eqref{rate2}, we have applied the property, concerning a positive random variable $x$, $\EE\left[ x\right]=\int_{t>0}\mathbb{P}\left( x>t \right)\mathrm{d}t$. Next, we consider that $h_{k}$ follows a Gamma distribution with shape and rate parameters, $\Delta$ and $1$, respectively. 
\end{appendices}

\bibliographystyle{IEEEtran}

\bibliography{mybib}
\begin{IEEEbiography}[{\includegraphics[width=1in,height=1.25in,clip,keepaspectratio]{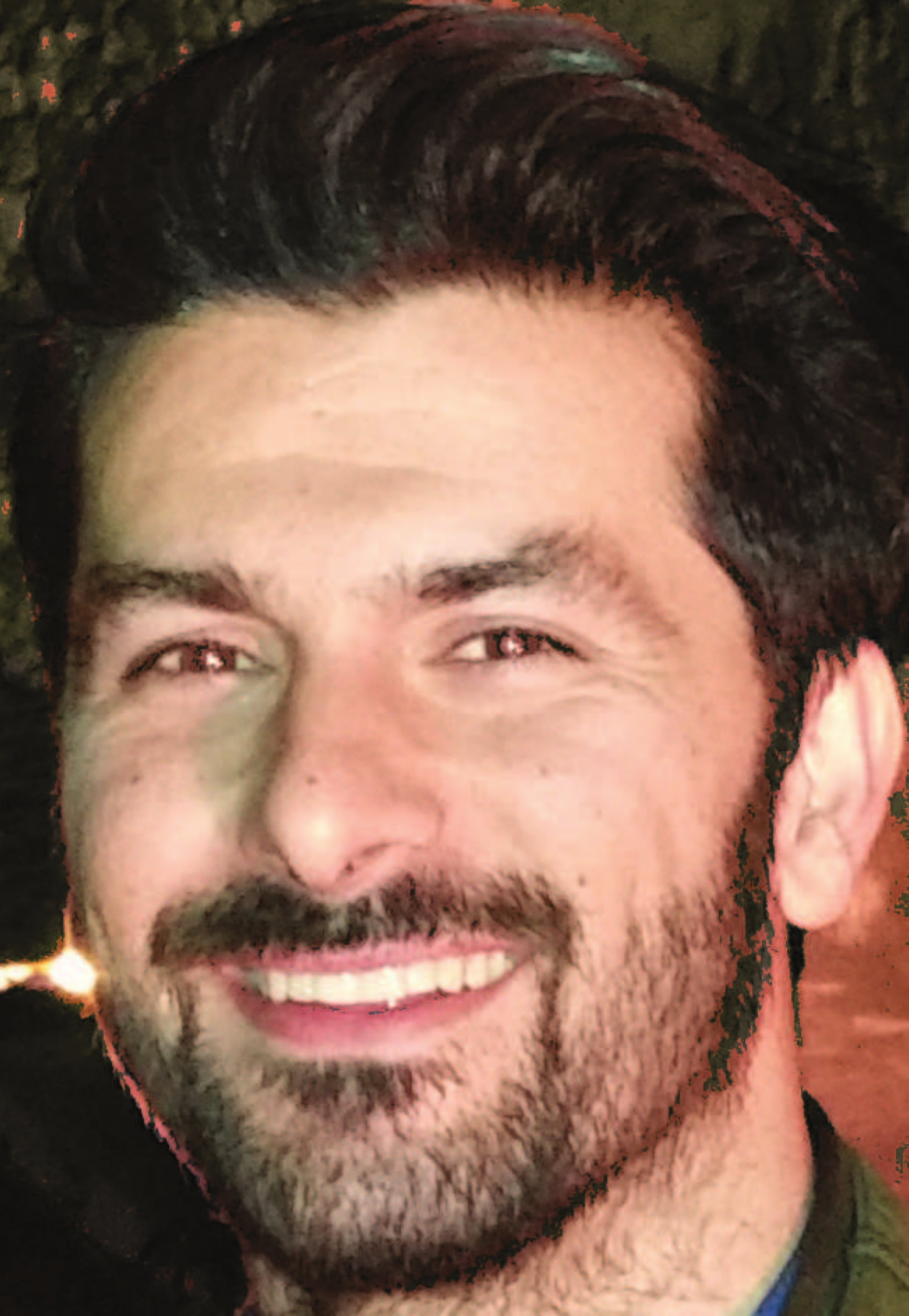}}]{Anastasios Papazafeiropoulos}[S'06-M'10] is currently a Research Fellow in IDCOM at the University of Edinburgh, U.K. He obtained the B.Sc in Physics and the M.Sc. in Electronics and Computers science both with distinction from the University of Patras, Greece in 2003 and 2005, respectively. He then received the Ph.D. degree from the same university in 2010. From November 2011 through December 2012 he was with the Institute for Digital Communications (IDCOM) at the University of Edinburgh, U.K. working as a postdoctoral Research Fellow, while during 2012-2014 he was a Marie Curie Fellow at Imperial College London, U.K.   Dr. Papazafeiropoulos has been involved in several EPSCRC and EU FP7 HIATUS and HARP projects. His research interests span massive MIMO, 5G wireless networks, full-duplex radio, mmWave communications, random matrices theory, signal processing for wireless communications, hardware-constrained communications,
and performance analysis of fading channels. 
\end{IEEEbiography}
\begin{IEEEbiography}
[{\includegraphics[width=1in,height=1.25in,clip,keepaspectratio]{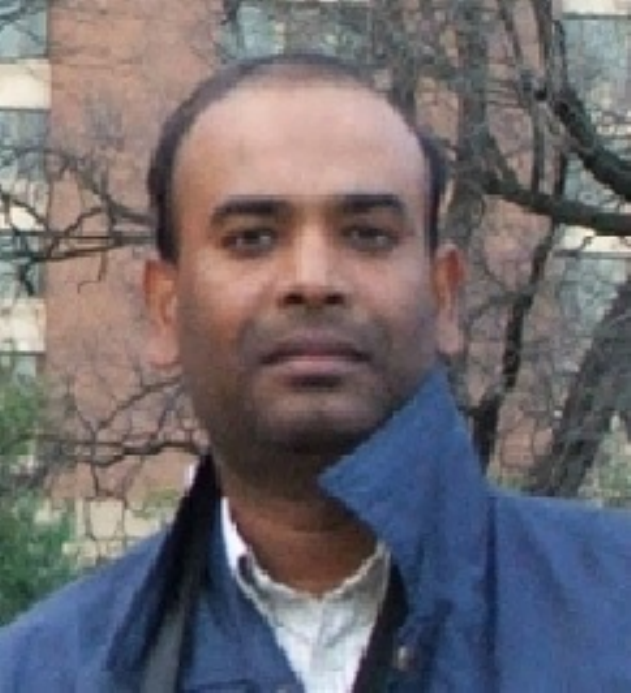}}]{Tharmalingam Ratnarajah}[A'96-M'05-SM'05] is currently with the Institute for Digital Communications, University of Edinburgh, Edinburgh, UK, as a Professor in Digital Communications and Signal Processing and the Head of Institute for Digital Communications. His research interests include signal processing and information theoretic aspects of 5G and beyond wireless networks, full-duplex radio, mmWave communications, random matrices theory, interference alignment, statistical and array signal processing and quantum information theory. He has published over 300 publications in these areas and holds four U.S. patents. 
He was the coordinator of the FP7 projects ADEL (3.7M\euro) in the area of licensed shared access for 5G wireless networks and HARP (3.2M\euro) in the area of highly distributed MIMO and FP7 Future and Emerging Technologies projects HIATUS (2.7M\euro) in the area of interference alignment and CROWN (2.3M\euro) in the area of cognitive radio networks. Dr Ratnarajah is a Fellow of Higher Education Academy (FHEA), U.K..
\end{IEEEbiography}

\end{document}